\documentclass[onecolumn,journal,letterpaper]{IEEEtran}
\usepackage{hyperref}
\pdfoutput=1

\usepackage{amssymb,amsmath,color,theorem}
\usepackage{graphicx}
\usepackage{algorithm,algorithmic,xspace}

\renewcommand{\natural}{{\mathbb{N}}}

\newcommand{\integernonnegative}{{\mathbb{Z}_{\geq0}}}
\renewcommand{\natural}{{\mathbb{Z}_{>0}}}
\newcommand{\naturalzero}{\mathbb{N}_0}
\renewcommand{\naturalzero}{\integernonnegative}

\newcommand{\real}{{\mathbb{R}}}
\newcommand{\subscr}[2]{{#1}_{\textup{#2}}}
\newcommand{\union}{\cup}
\newcommand{\intersection}{\ensuremath{\operatorname{\cap}}}

\newcommand{\map}[3]{#1: #2 \rightarrow #3}

\newcommand{\setdef}[2]{\{#1 \; | \; #2\}}

\newcommand{\norm}[1]{\|#1\|}

\newcommand{\diam}{\operatorname{diam}}
\newcommand{\radius}{\operatorname{radius}}

\newcommand{\alphabet}{\mathbb{A}}
\newcommand{\algofont}[1]{{\small\textsc{#1}}}

\newcommand{\e}{e}
\newcommand{\Enorm}[1]{\|#1\|_{2}}

\newcommand{\umax}{\subscr{r}{ctr}}
\renewcommand{\umax}{\subscr{u}{max}}
\newcommand{\rcmm}{\subscr{r}{cmm}}
\newcommand{\vmax}{\subscr{v}{max}}

\newcommand{\xmin}{\subscr{x}{min}}
\newcommand{\xmax}{\subscr{x}{max}}
\newcommand{\ymin}{\subscr{y}{min}}
\newcommand{\ymax}{\subscr{y}{max}}

\newcommand{\Hoptimal}[1]{H^{[#1]}_{\textup{optimal}}}
\newcommand{\Hmeasured}[1]{H^{[#1]}_{\textup{measured}}}
\renewcommand{\Hmeasured}[1]{H^{[#1]}_{\textup{sensed}}}

\newcommand{\GG}{\mathcal{G}}
\newcommand{\cball}[2]{B(#2,#1)}

\newcommand{\Bigcball}[2]{B\Big(#2,#1\Big)}
\newcommand{\bigcball}[2]{B\big(#2,#1\big)}
\newcommand{\card}{\ensuremath{\operatorname{card}}}
\newcommand{\until}[1]{\{1,\dots,#1\}}
\newcommand{\fromto}[2]{\{#1,\dots,#2\}}

\newcommand{\HH}{\mathcal{H}}

\newcommand{\XX}{\mathcal{X}}

\newcommand{\zeroRd}{{0}_{\real^d}}
\newcommand{\oneRd}{{1}_{\real^d}}
\renewcommand{\zeroRd}{{0}_d}
\renewcommand{\oneRd}{{1}_d}

\newcommand{\subject}{\text{subject to}}

\newcommand{\half}{\frac{1}{2}}

\newcommand{\ViolTest}[2]{\textup{\texttt{Viol}}(#1, #2)}
\newcommand{\Basis}[2]{\textup{\texttt{Basis}}(#1, #2)}
\newcommand{\subexLP}{\textup{\texttt{SUBEX\_lp}}}
\renewcommand{\subexLP}{\textup{\texttt{Subex\_LP}}}
\newcommand{\NALPmap}{\mathcal{B}}

\newcommand{\outnbrs}{\subscr{\mathcal{N}}{out}}
\newcommand{\innbrs}{\subscr{\mathcal{N}}{in}}

\newcommand{\dist}{\operatorname{dist}}

\newcommand{\fti}{\text{fti}}

\newcommand{\supind}[2]{{#1}^{[#2]}}
\newcommand{\subind}[2]{{#1}_{#2}}

\newcommand{\nll}{\textup{\texttt{null}}\xspace}
\newcommand{\ctrl}{\textup{ctl}}
\newcommand{\msg}{\textup{msg}}

\newcommand{\stf}{\textup{stf}}

\newtheorem{theorem}{Theorem}[section]
\newtheorem{proposition}[theorem]{Proposition}

\newtheorem{definition}[theorem]{Definition}
\newtheorem{lemma}[theorem]{Lemma}
{\theorembodyfont{\rmfamily} 
\newtheorem{conjecture}[theorem]{Conjecture}
\newtheorem{remark}[theorem]{Remark}

\newtheorem{example}[theorem]{Example}

\newtheorem{problem}[theorem]{Problem}}

\newcommand\oprocendsymbol{\hbox{$\square$}}
\newcommand\oprocend{\relax\ifmmode\else\unskip\hfill\fi\oprocendsymbol}

\renewcommand{\theenumi}{(\roman{enumi})}

\renewcommand{\baselinestretch}{1.05}

\newcommand{\PiLP}{\Pi_{\textup{LP}}}

\graphicspath{{figs/}}

\begin{document}

\title{Distributed Abstract Optimization via Constraints Consensus: Theory
  and Applications\thanks{This material is based upon work supported in
    part by ARO MURI Award W911NF-05-1-0219, ONR Award N00014-07-1-0721,
    and NSF Award CNS-0834446. The research leading to these results has
    received funding from the European Community's Seventh Framework
    Programme (FP7/2007-2013) under grant agreement no. 224428 (CHAT
    Project). The authors would like to thank Dr.\ Colin Jones for helpful
    comments.
Early short versions of this work appeared
as~\cite{GN-FB:06d,GN-FB:06z,GN-FB:08x}: differences between these early
short versions and the current article include a much improved
comprehensive treatment, revised complete proofs for all statements, and
the Monte Carlo analysis.
}}

\author{Giuseppe Notarstefano\thanks{Giuseppe Notarstefano is with the
    Department of Engineering, University of Lecce, Via per Monteroni,
    73100 Lecce, Italy, \texttt{giuseppe.notarstefano@unile.it}}
  \quad\and\quad
  Francesco Bullo\thanks{Francesco Bullo is with the Center for Control,
    Dynamical Systems and Computation, University of California at Santa
    Barbara, CA 93106, USA, \texttt{bullo@engineering.ucsb.edu}} }

\maketitle

\begin{abstract}
  Distributed abstract programs are a novel class of distributed
  optimization problems where (i) the number of variables is much smaller
  than the number of constraints and (ii) each constraint is associated to
  a network node. Abstract optimization programs are a generalization of
  linear programs that captures numerous geometric optimization problems.
  We propose novel constraints consensus algorithms for distributed
  abstract programs: as each node iteratively identifies locally active
  constraints and exchanges them with its neighbors, the network computes
  the active constraints determining the global optimum.
  The proposed algorithms are appropriate for networks with weak
  time-dependent connectivity requirements and tight memory constraints.
  We show how suitable target localization and formation control problems
  can be tackled via constraints consensus.
\end{abstract}


\section{Introduction}
This paper focuses on a class of distributed optimization problems and its
application to target localization and formation control.
Distributed optimization and computation have recently received widespread
attention in the context of distributed estimation in sensor networks,
distributed control of actuator networks and consensus algorithms. An early
established reference on distributed optimization is \cite{JNT-DPB-MA:86},
whereas a non-exhaustive set of recent references includes
\cite{NM-AJ:09,AN-AO-PAP:09,BJ-AS-MK-KHJ:08,DPP-MC:07}.
We consider a distributed version of abstract optimization problems.
Abstract optimization problems, sometimes referred to as abstract linear
programs or as LP-type programs, generalize linear programming and model a
variety of machine learning and geometric optimization problems. Examples
of geometric optimization problems include the smallest enclosing ball, the
smallest enclosing stripe and the smallest enclosing annulus problems.
Early references to abstract optimization problems
include~\cite{JM-MS-EW:96,BG:95,MG:95}.
In this paper we are interested in abstract optimization problems where the
number of constraints $n$ is much greater than the number of constraints
$\delta$ that identify the optimum solution (and where, therefore, there is
a large number of redundant constraints). For example, we are interested in
linear programs where $n$ is much greater than the number of variables $d$
(in linear programs, $\delta=d$). Under this dimensionality assumption,
we consider distributed versions of abstract optimization programs, where
$n$ is also the number of network nodes and where each constraint is
associated to a node. We consider processor networks described by
arbitrary, possibly time-dependent communication topologies and by
computing nodes with tight memory constraints. After presenting and
analyzing constraints consensus algorithms for distributed abstract
optimization, we apply them to target localization in sensor networks and
to formation control in robotic networks.

The relevant literature is vast; we organize it in three broad
areas. First, linear programming and its generalizations, including
abstract optimization, have received widespread attention in the
literature. For linear programs in a fixed number of variables subject to
$n$ linear inequalities, the earliest algorithm with time complexity in
$O(n)$ is given in~\cite{NM:84}.  An efficient randomized algorithm is
proposed in~\cite{JM-MS-EW:96}, where a linear program in $d$ variables
subject to $n$ linear inequalities is solved in expected time $O(d^2 n +
\e^{O(\sqrt{d\log d})})$; the expectation is taken over the internal
randomizations executed by the algorithm.  An elegant survey on randomized
methods in linear programming and on abstract optimization
is~\cite{BG-EW:96}; see also~\cite{MG:95,BG-EW:01}.  The
survey~\cite{PKA-SS:01}, see also \cite{PKA-MS:98}, discusses the
application of abstract optimization to a number of geometric optimization
problems.  Regarding parallel computation approaches to linear programming,
linear programs with $n$ linear inequalities can be solved~\cite{MA-NM:96}
by $n$ parallel processors in time $O((\log\log(n))^d)$.  However, the
approach in~\cite{MA-NM:96}, see also references therein, is limited to
parallel random-access machines, where a shared memory is readable and
writable to all processors.  Other references on distributed linear
programming include~\cite{YB-JWB-DR:04,HD-HK:08}.

A second relevant literature area is distributed training of support vector
machines (SVMs).  A randomized parallel algorithm for SVM training is
proposed in~\cite{YL-VR:06} by using methods from abstract optimization and
by exploiting the idea of exchanging only active constraints. Along these
lines, \cite{YL-VR-LV:08} extends the algorithm to parallel computing over
strongly connected networks, \cite{JB-YD-JT-OW:08} contains a comprehensive
discussion of SVM training via abstract optimization, and
\cite{KF-BBL-PT:06} applies similar algorithmic ideas to wireless sensor
networks.  The algorithms in~\cite{YL-VR:06,YL-VR-LV:08}, independently
developed at the same time of our
works~\cite{GN-FB:06d,GN-FB:06z,GN-FB:08x}, differ from our constraint
consensus algorithm in the following ways.  First, the number of
constraints stored at the nodes grows at each iteration so that both the
memory and the local computation time at each node may be of order $n$.
Second, our algorithm is proposed for general abstract optimization
problems and thus may be applied to a variety of application
domains. Third, our algorithm exploits a novel re-examination idea, is
shown to be correct for time-varying (jointly strongly connected) digraphs,
and features a distributed halting condition.

As third and final set of relevant references, we include a brief synopsis
of recent progress in target localization in sensor networks and formation
control in robotic networks.  The problem of target localization has been
widely investigated and recent interest has focused on sensors and wireless
networks; e.g., see the recent text~\cite{FZ-LG:04}. In this paper we take
a deterministic worst-case approach to localization, adopting the set
membership estimation technique proposed in~\cite{AG-AV:01}. A related
sensor selection problem for target tracking is studied in \cite{VI-RB:06}.
Regarding the literature on formation control for robotic networks, an
early reference on distributed algorithms and geometric patterns is
\cite{IS-MY:99}.  Regarding the rendezvous problem, that is, the problem of
gathering the robots at a common location, an early reference
is~\cite{HA-YO-IS-MY:99}.  The ``circle formation control'' problem, i.e.,
the problem of steering the robots to a circle formation, is discussed
in~\cite{XD-AK:02}.  The
references~\cite{ME-XH:01b,HGT-GJP-VK:04,JAM-MEB-BAF:04c} are based on,
respectively, control-Lyapunov functions, input-to-state stability and
cyclic pursuit.

The contributions of this paper are twofold.  First, we identify and study
distributed abstract programming as a novel class of distributed
optimization problems that are tractable and widely applicable.  We propose
a novel algorithmic methodology, termed \emph{constraints consensus}, to
solve these problems in networks with various connectivity and memory
constraints: as each node iteratively identifies locally active constraints
and exchanges them with its neighbors, the globally active constraints
determining the global optimum are collectively identified. A constraint
re-examination idea is the distinctive detail of our algorithmic design.
We propose three algorithms, a nominal one and two variations, to solve
abstract programs depending on topology, memory and computation
capabilities of the processor network. 
We formally establish various algorithm properties, including monotonicity,
finite-time convergence to consensus, and convergence to the
possibly-unique correct solution of the abstract program.  Moreover, we
provide a distributed halting conditions for the nominal algorithm.
We provide a conservative upper bound on the completion time of the nominal
algorithm and conjecture that the completion time depends linearly on $n$
(i.e., the number of constraints and the network dimension).  Next, we
evaluate the algorithm performance via a Monte Carlo probability-estimation
analysis and we substantiate our conjecture on stochastically-generated
sample problems. Sample problems are randomly generated by considering two
classes of linear programs, taken from~\cite{RS:87}, and three types of
graphs (line-graph, Erd\H{o}s-R\`enyi random graph and random geometric
graph).

As a second set of contributions, we illustrate how distributed abstract
programs are relevant for distributed target localization in sensor
networks and for formation control problems, such as the rendezvous problem
and the line or circle formation problems. Specifically, for the target
localization problem, we design a distributed algorithm to estimate a
convex polytope, specifically an axis-aligned bounding box, containing the
moving target.  Our proposed \emph{eight half-planes consensus algorithm}
combines (i) distributed linear programs to estimate the convex polytope at
a given instant and (ii) a set-membership recursion, consisting of
prediction and update steps, to dynamically track the region. We discuss
correctness and memory complexity of the distributed estimation algorithm.
Next, regarding formation control problems, we design a joint communication
and motion coordination scheme for a robotic networks model involving
range-based communication.  We consider formations characterized by the
geometric shapes of a point, a line, or a circle.  We solve these formation
control problems in a time-efficient distributed manner combining two
algorithmic ideas: (i) the robots implement a constraints consensus
algorithm to compute a common shape reachable in minimum-time, and (ii) the
network connectivity is maintained by means of an appropriate standard
connectivity-maintenance strategy. 
In the limit of vanishing robot displacement per communication round, our
proposed \emph{move-to-consensus-shape} strategy solves the optimal
formation control tasks.

\subsubsection*{Paper organization}
The paper is organized as follows.  Section~\ref{sec:LP-type problems}
introduces abstract optimization problems.
Section~\ref{sec:network-modeling} introduces network models.
Section~\ref{sec:network-ALP} contains the definition of distributed
abstract program and the constraints consensus
algorithms. Section~\ref{sec:computations} contains the Monte Carlo
analysis of the time-complexity of the constraints consensus
algorithm. Sections~\ref{sec:target-localization}
and~\ref{sec:mintime-formation} contain the application of the proposed
constraints consensus algorithms to target localization and formation
control.

\subsubsection*{Notation}
We let $\natural$, $\naturalzero$, and $\real_{> 0}$ denote the natural
numbers, the non-negative integer numbers, and the positive real numbers,
respectively. For $r\in\real_{> 0}$ and $p\in\real^d$, we let
$\cball{r}{p}$ denote the closed ball centered at $p$ with radius $r$, that
is, $\cball{r}{p}=\setdef{q\in\real^d}{\Enorm{p-q}\leq{r}}$.  For
$d\in\natural$, we $\zeroRd$ and $\oneRd$ denote the vectors in $\real^d$
whose entries are all $0$ and $1$, respectively. Similarly, we let
$\infty_d$ and $-\infty_d$ the vectors with $d$ entries $\infty$ and
$-\infty$, respectively.  For a finite set $A$, we let $\card(A)$ denote
its cardinality.  For two functions $f,g:\natural\to\real_{> 0}$, we write
$f(n) \in O(g)$ (respectively, $f(n)\in\Omega(g)$) if there exist
$N\in\natural$ and $c\in\real_{>0}$ such that $f(n) \leq cg(n)$ for all $n
\geq N$ (respectively, $f(n) \geq cg(n)$ for all $n \geq N$).  Finally, we
introduce the convention that sets are allowed to contain multiple copies
of the same element.

Given $S\subset\real^d$ and $p\in\real^d$, let $\dist(p, S)$ denote the
distance from $p$ to $S$, that is, $\dist(p,S)=\inf_{s\in S}\norm{p-s}{}$.
For distinct $p_1\in\real^d$ and $p_2\in\real^d$, let $\ell(p_1,p_2)$ be
the line through $p_1$ and $p_2$.
In what follows, a set of distinct points
$\{p_1,\dots,p_n\}\subset\real^d$, $n\geq3$, is in \emph{stripe-generic
  position} if, given any two ordered subsets $(p_a,p_b,p_c)$ and
$(p_\alpha,p_\beta,p_\gamma)$, either $\dist(p_a, \ell(p_b,p_c)) \neq
\dist(p_\alpha, \ell(p_\beta,p_\gamma))$ or
$(p_a,p_b,p_c)=(p_\alpha,p_\beta,p_\gamma)$.

\section{Abstract optimization}
\label{sec:LP-type problems}
In this section we present an abstract framework~\cite{PKA-SS:01,BG-EW:01}
that captures a wide class of optimization problems including linear
programming and various machine learning and geometric optimization
problems. Abstract optimization problems are also known as \emph{abstract
  linear programs}, \emph{generalized linear programs} or \emph{LP-type
  problems}.

\subsection{Problem setup and examples }
We consider optimization problems specified by a pair $(H, \phi)$, where
$H$ is a finite set, and $\map{\phi}{2^H}{\Phi}$ is a
function\footnote{Given a set $H$, the set $2^H$ is the set of all subsets
  of $H$} with values in a linearly ordered set ($\Phi, \leq$); we assume
that $\Phi$ has a minimum value $-\infty$. The elements of $H$ are called
\emph{constraints}, and for $G\subset H$, $\phi(G)$ is called the
\emph{value} of $G$. Intuitively, $\phi(G)$ is the smallest value
attainable by a certain objective function while satisfying the constraints
of $G$. An optimization problem of this sort is called an \emph{abstract
  optimization program} if the following two axioms are satisfied:
\begin{enumerate}
\item \emph{Monotonicity}: if $F\subset G \subset H$, then $\phi(F) \leq
  \phi(G)$;
\item \emph{Locality}: if $F \subset G \subset H$ with $-\infty < \phi(F) =
  \phi(G)$, then, for all $h\in H$,
  \begin{equation*}
    \phi(G) < \phi(G \union \{h\})
    \enspace\implies\enspace \phi(F) < \phi(F \union \{h\}).
\end{equation*}
\end{enumerate}
A set $B\subset H$ is \emph{minimal} if $\phi(B)>\phi(B')$ for all
proper subsets $B'$ of $B$.  A minimal set $B$ with $-\infty<\phi(B)$ is
a $\emph{basis}$.  Given $G \subset H$, a \emph{basis of $G$} is a minimal
subset $B \subset G$, such that $-\infty<\phi(B) = \phi(G)$.  A
constraint $h$ is said to be \emph{violated} by $G$, if $\phi(G) <
\phi(G \union \{h\})$.

A \emph{solution} of an abstract optimization program $(H,\phi)$ is a
minimal set $B_H\subset H$ with the property that $\phi(B_H) = \phi(H)$.
The \emph{combinatorial dimension} $\delta$ of $(H, \phi)$ is the maximum
cardinality of any basis. Finally, an abstract program is called
\emph{basis regular} if, for any basis with $\card(B)=\delta$ and any
constraint $h\in H$, every basis of $B \union \{h\}$ has the same
cardinality of $B$. We now define two important primitive operations that
are useful to solve abstract optimization problems:
\begin{enumerate}
\item \emph{Violation test}: given a constraint $h$ and a basis $B$, it
  tests whether $h$ is violated by $B$; we denote this primitive by
  $\ViolTest{B}{h}$;
\item \emph{Basis computation}: given a constraint $h$ and a basis $B$, it
  computes a basis of $B\union \{h\}$;  we denote this primitive by
  $\Basis{B}{h}$.
\end{enumerate}

\begin{example}[Abstract framework for linear programs]
  \label{ex:abstract-framework-for-LP}
  We recall from \cite{BG-EW:96} how to transcribe a linear program into an
  abstract optimization program.  A linear program (LP) in ${x\in\real^d}$
  is given by
  \begin{equation*}
    \begin{split}
      \min &\quad c^T x\\
      \subject &\quad a_i^T x \leq b_i, \quad i\in\until{n}
    \end{split}
  \end{equation*}
  where $d\in \natural$ is the state dimension, $c\in\real^d$ characterizes
  the linear cost function to minimize, and $a_i \in \real^{d}$ and $b_i\in
  \real$ describe $n\in\natural$ inequality constraints.  In order to
  transcribe the LP into an abstract program, we need to specify the
  constraint set $H$ and the value $\phi(G)$ for each $G\subset H$. The
  constraint set $H$ is simply the set of half-spaces $\{h_1, \ldots,
  h_n\}$, where $h_i = \setdef{x\in\real^d}{a_i^T x \leq b_i}$.  Defining
  the value function in order to satisfy the monotonicity and locality
  axioms is more delicate: if $\Phi=(\real,\leq)$ and $\phi(G)$ is the
  minimum of $c^T x$ subject to the constraint set $G$, then the locality
  axiom no longer holds (see Section~4 in \cite{BG-EW:96} for a
  counterexample).  A correct choice is as follows: let $(\Phi, \leq)$ be
  the set $\real^d$ with the \emph{lexicographical order},\footnote{In the
    lexicographic order on $\real^2$, we have $(x_1,y_1)\leq(x_2,y_2)$ if
    and only if $x_1<x_2$ or ($x_1=x_2$ and $y_1\leq y_2$).} and define
  $\phi(G) = v_G$, where $v_G \in \real^d$ is the (unique)
  \emph{lexicographically minimal} point $x$ minimizing $c^T x$ over the
  constraint set $G$, when it exists and is bounded. If the problem is
  infeasible (the intersection of the constraints in $G$ is empty), then
  $v_G = \infty_d$.  If the problem is unbounded (no lexicographically
  minimal point exists), then $v_G = -\infty_d$.  If $v_G$ is finite, then
  a basis of $G$ is a minimal subset of constraints $B\subset G$ such that
  $v_B = v_G$.  It is known~\cite{JM-MS-EW:96} that the abstract
  optimization program transcription of a feasible LP is basis regular and
  has combinatorial dimension $d$. A constraint $h\in H$ is violated by $G$
  if and only if $v_G < v_{G\union \{h\}}$.  \oprocend
\end{example}


\begin{example}\textbf{\textup{(Abstract optimization problems in
    geometric optimization)}}
  \label{rem:examples}
  We present three useful geometric examples, illustrated in
  Figure~\ref{fig:geometric-alp}.
  \begin{enumerate}
  \item \emph{Smallest enclosing ball:} Given $n$ distinct points
    in~$\real^d$, compute the center and radius of the ball of smallest
    volume containing all the points.  This problem is~\cite{JM-MS-EW:96}
    an abstract optimization program with combinatorial dimension~$d+1$.

  \item \emph{Smallest enclosing stripe:}
    Given $n$ distinct points in~$\real^2$ in stripe-generic positions,
    compute the center and the width of the stripe of smallest width
    containing all the points. In the Appendix we prove (for the first time
    at the best of our knowledge) that this problem is an abstract
    optimization program with combinatorial dimension~$5$.

  \item \emph{Smallest enclosing annulus:} Given $n$ distinct points
    in~$\real^2$, compute the center and the two radiuses of the annulus of
    smallest area containing all the points.  This problem
    is~\cite{JM-MS-EW:96} an abstract optimization program with
    combinatorial dimension~$4$.
  \end{enumerate}
  \begin{figure}[htbp]
    \centering
    \includegraphics[height=.5\linewidth,angle=90]{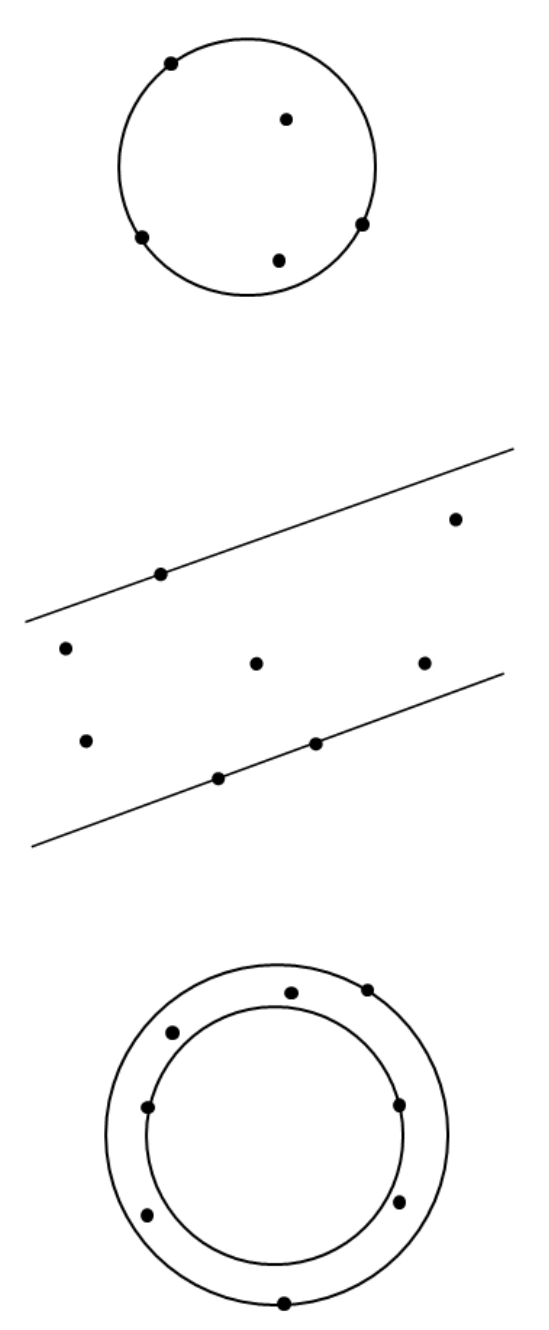}%
    \caption{Smallest enclosing ball, stripe and annulus}
    \label{fig:geometric-alp}
  \end{figure}
  In all three examples, the violation test and the basis computation
  primitives amount to low-dimensional geometric problems and are more or
  less straightforward. Numerous other geometric optimization problems can
  be cast as abstract optimization programs as discussed
  in~\cite{JM-MS-EW:96,BG:95,BG-EW:96,PKA-SS:01} \oprocend
\end{example}

We end this section with a useful lemma, that is an immediate consequence
of locality, and a useful rare property.
\begin{lemma}
  \label{lemma:locality2} For any $F$ and $G$ subsets of $H$,
  $\phi(F\union{G})>\phi(F)$ if and only if there exists $g\in{G}$ such that
  $\phi(F\union\{g\})>\phi(F)$.
\end{lemma}
\begin{proof}
  If there exists $g\in{G}$ such that $\phi(F\union\{g\})>\phi(F)$, then by
  monotonicity $\phi(F\union{G})\geq \phi(F\union\{g\})>\phi(F)$. For the
  other implication, assume $G=\{g_1,\ldots,g_k\}$ for some $k\in\natural$,
  and define $G_i := \{g_1,\ldots,g_i\}$ for $i\in\until{k}$. We may
  rewrite the assumption $\phi(F\union{G})>\phi(F)$ as $\phi(F\union
  G_{k-1}\union \{g_k\})>\phi(F)$. If $\phi(F\union G_{k-1})=\phi(F)$, then
  the locality axiom implies $\phi(F\union \{g_k\})>\phi(F)$ and the thesis
  follows with $g=g_k$. Otherwise, the same argument may be applied to
  $G_{k-1}$. The recursion stops either when $\phi(F\union{G_{i}})=\phi(F)$
  (and the thesis follows with $g=g_{i+1}$) for some $i$ or when
  $\phi(F\union G_{1})>\phi(F)$ (and the thesis follows with $g=g_1$).
\end{proof}

Next, given an abstract optimization program $(H,\omega)$, let $B_G$ denote
the basis of any $G\subseteq H$.  An element $h$ of $B_H$ is
\emph{persistent} if $h\in B_G$ for all $G\subseteq H$ containing $h$.  An
abstract optimization program $(H,\omega)$ is \emph{persistent} if all
elements of $B_H$ are persistent. The persistence property is useful, as we
state in the following result.
\begin{lemma}
  \label{lemma:all-is-simple-if-persistent}
  Any persistent abstract optimization program $(H,\omega)$ can be solved in a
  number of time steps equal to the dimension of $H$.
\end{lemma}
\begin{proof}
  Let $H = \{h_1, \ldots, h_n\}$. Set $B = \{h_1, \ldots, h_\delta\}$ and
  then update $B = \Basis{B}{h_k}$ for $k = \delta+1, \ldots, n$. Because
  of persistency, each $h\in B_H$ is added to $B$ once it is selected as
  $h_k$ and is not removed from $B$ in subsequent basis computations.
\end{proof}
Unfortunately, the persistence property is rare.  Indeed, in
Figure~\ref{fig:LP_counter_example} we show an LP problem where the
persistency property does not hold. In fact, it can be easily noticed that
$\{h_1,h_2\}$ is a basis for $\{h_1,h_2,h_3,h_4\}$, but $\{h_3, h_4\}$ is a
basis for $\{h_2,h_3,h_4\}$. In other words $\{h_2\}$ is not violated by
$\{h_3,h_4\}$.  The lack of persistency complicates the design of
centralized and distributed solvers for abstract optimization problems. For
example, in network settings, flooding algorithms are not sufficient.

\begin{figure}[htbp]
  \centering
  \includegraphics[width=.3\linewidth]{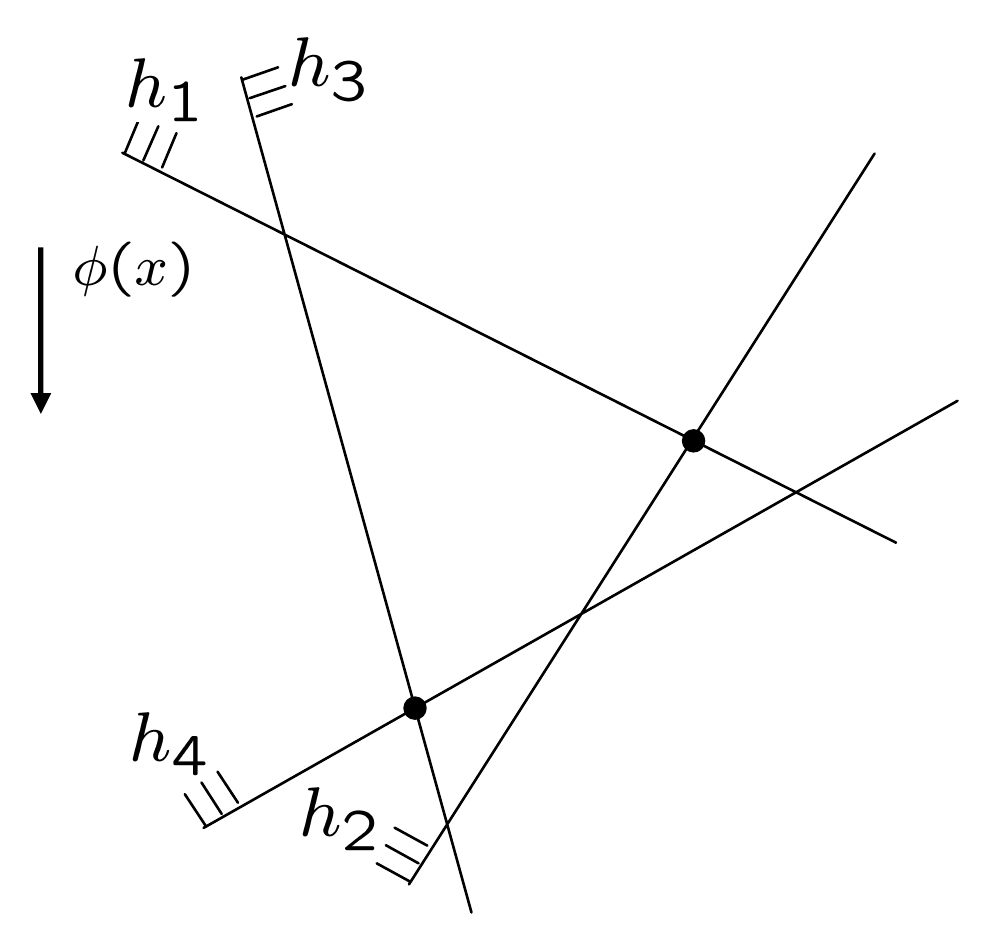}
  \caption{An LP problem that is not persistent.}
  \label{fig:LP_counter_example}
\end{figure}

\subsection{Randomized sub-exponential algorithm}
In the following sections we will assume that each node in the network
possesses a routine capable of solving small-dimensional abstract
optimization programs. For completeness' sake, this section reviews the
randomized algorithm proposed in \cite{JM-MS-EW:96}. This algorithm has
expected running time with linear dependence on the number of constraints,
whenever the combinatorial dimension $\delta$ is fixed, and with
sub-exponential dependence on the $\delta$; these bounds are proven in
\cite{JM-MS-EW:96} for linear programs and in \cite{BG-EW:96} for general
abstract optimization programs.
The algorithm, called $\subexLP$, has a recursive structure and is based on
the violation test and the basis computation primitives.  Given a set of
constraints $G$ and a candidate basis $C \subset G$, the algorithm is
stated as follows:
\begin{center}
\begin{minipage}[c]{.9\textwidth}
\textbf{function} $\subexLP(G, C)$
\begin{algorithmic}[1]
\STATE \textbf{if} $G = C$, \textbf{then} \textbf{return} $C$
\STATE \textbf{else}
\STATE \quad choose a random $h\in G \setminus C$ and
compute $B := \subexLP ( G\setminus\{h\}, C)$
\STATE \quad \textbf{if} {$\ViolTest{B}{h}$  (that is, $h$ is violated by
  $B$),} \textbf{then}
\STATE \quad \quad compute $B' := $ basis for $B\union\{h\}$
\STATE \quad \quad \textbf{return} $\subexLP(G,\Basis{B}{h})$
\STATE \quad \textbf{else} \textbf{return} $B$
\STATE \quad \textbf{endif}
\STATE \textbf{endif}
\end{algorithmic}
\end{minipage}
\end{center}

\noindent For the abstract optimization program $(H,\phi)$, the routine is
invoked with $\subexLP(H,B)$, given any initial candidate basis $B$.

\section{Network models}
\label{sec:network-modeling}
Following \cite{NAL:97}, we define a synchronous network system as a
``collection of computing elements located at nodes of a directed network
graph.''  We refer to computing elements are processors.

\subsection{Digraphs and connectivity}
We let $\GG = (\until{n}, E)$ denote a directed graph (or digraph), where
$\until{n}$ is the set of nodes and $E\subset \until{n}^2$ is the set of
edges.  For each node $i$ of $\GG$, the number of edges going out from
(resp. coming into) node $i$ is called \emph{out-degree}
(resp. \emph{in-degree}).
A digraph is \emph{strongly connected} if, for every pair of nodes $(i, j)
\in \until{n}\times\until{n}$, there exists a path of directed edges that
goes from $i$ to $j$. A digraph is \emph{weakly connected} if replacing all
its directed edges with undirected edges results in a connected undirected
graph. In a strongly connected digraph, the minimum number of edges between
node $i$ and $j$ is called the \emph{distance from $i$ to $j$} and is
denoted $\dist(i,j)$. The maximum $\dist(i,j)$ taken over all pairs $(i,j)$
is the \emph{diameter} and is denoted $\diam(\GG)$.  Finally, we consider
time-dependent digraphs of the form $t\mapsto \GG(t) = (\until{n},
E(t))$. The time-dependent digraph $\GG$ is \emph{jointly strongly
  connected} if, for every $t \in \naturalzero$, the digraph
$\union_{\tau=t}^{+\infty} \GG(\tau)$ is strongly connected.

In a time-dependent digraph, the set of outgoing (incoming) neighbors of
node $i$ at time $t$ are the set of nodes to (from) which there are edges
from (to) $i$ at time $t$. They are denoted by $\outnbrs(i,t)$ and
$\innbrs(i,t)$, respectively.

\subsection{Synchronous networks and distributed algorithms}
\label{subsec:networkmodel+distributedalgos}
A \emph{synchronous network} is a time-dependent digraph $\GG = (\until{n},
\subscr{E}{cmm})$, where $\until{n}$ is the set of \emph{identifiers} of
the processors, and the time-dependent set of edges $\subscr{E}{cmm}$
describes communication among processors as follows: $(i,j)$ is in
$\subscr{E}{cmm}(t)$ if and only if processor $i$ can communicate to
processor $j$ at time $t\in\integernonnegative$.

For a synchronous network $\GG$ with processors $\until{n}$, a
\emph{distributed algorithm} consists of (1) the set $W$, called the set of
\emph{processor states} $\supind{w}{i}$, for all $i\in\until{n}$; (2) the
set $\alphabet$, called the \emph{message alphabet}, including the $\nll$
symbol; (3) the map $\map{\msg}{W \times\until{n}}{\alphabet}$, called the
\emph{message-generation function}; and (4) the map $\map{\stf}{W
  \times\alphabet^n}{W}$, called the \emph{state-transition function}.
The execution of the distributed algorithm by the network begins with all
processors in their start states.  The processors repeatedly perform the
following two actions. First, the $i$th processor sends to each of its
outgoing neighbors in the communication graph a message (possibly the
$\nll$ message) computed by applying the message-generation function to the
current value of $\supind{w}{i}$.  After a negligible period of time, the
$i$th processor computes the new value of its processor state
$\supind{w}{i}$ by applying the state-transition function to the current
value of $\supind{w}{i}$, and to the incoming messages (present in each
communication edge). The combination of the two actions is called a
\emph{communication round} or simply a round.

In this execution scheme we have assumed that each processor executes all
the calculations in one round. If it is not possible to upper bound the
execution-time of the algorithm, then one may consider a slightly different
network model that allows the state-transition function to be executed
across multiple rounds. When this happens, the message is generated by
using the processor state at the previous round.

The last aspect to consider is the \emph{algorithm halting}, that is a
situation such that the network (and therefore each processor) is in a idle
mode.  Such status can be used to indicate the achievement of a prescribed
task.  Formally we say that a distributed algorithm is in halting status if
the processor state is a fixed point for the state-transition function
(that becomes a self-loop) and no message (or equivalently the $\nll$
message) is generated at each node.

\section{Distributed abstract optimization}
\label{sec:network-ALP}
In this section we define distributed abstract programs, propose novel
distributed algorithms for their solutions and analyze their correctness.

\subsection{Problem statement}
Informally, a \emph{distributed abstract program} consists of three main
elements: a network, an abstract optimization program and a mechanism to
distribute the constraints of the abstract program among the nodes of the
network.
\begin{definition}
  A distributed abstract program is a tuple $( \GG, (H, \phi),
  \NALPmap)$ consisting of
  \begin{enumerate}
  \item $\GG = (\until{n}, \subscr{E}{cmm})$, a synchronous network;
  \item $(H, \phi)$, an abstract program; and
  \item $\map{\NALPmap}{H}{\until{n}}$, a surjective map called
    \emph{constraint distribution map} that associates to each constraint
    one network node.
  \end{enumerate}
  If the map $\NALPmap$ is a bijection, we denote the distributed abstract
  program with the pair $(\GG, (H, \phi))$. A \emph{solution} of $( \GG,
  (H, \phi), \NALPmap)$ is attained when all network processors have
  computed a solution to $(H, \phi)$.
\end{definition}

\begin{remark}
  The most natural choice of constraint distribution map $\NALPmap$ is a
  bijection; in this case, (i) the network dimension is equal to the
  dimension of the abstract optimization program and (ii) precisely one
  constraint is assigned to each network node.  More complex distribution
  maps are interesting depending on the computation power and memory of the
  network processors.  In what follows, we typically assume $\NALPmap$ to
  be a bijection. \oprocend
\end{remark}

\subsection{Constraints consensus algorithms}
Here we propose three novel distributed algorithms that solve distributed
abstract programs.  First, we describe a distributed algorithm that
is well-suited for time-dependent networks whose nodes have bounded
computation time, memory and in-degree.  Equivalently, the algorithm is
applicable to networks with arbitrary in-degree, but also arbitrary
computation time and memory.
Then we describe two variations that deal with arbitrary in-degree versus
short computation time and small memory. The second version of the
algorithm is well-suited for time-dependent networks that have arbitrary
in-degree and bounded computation time, but also arbitrary memory (in the
sense that the number of stored messages may depend on the number of nodes
of the network).  The third algorithm considers the case of
time-independent networks with arbitrary in-degree and bounded computation
time and memory.

In all algorithms we consider a synchronous network $\GG$ and an abstract
program $(H, \phi)$ with $H = \{h_1, \cdots, h_n\}$ and with combinatorial
dimension $\delta$.  We define a distributed abstract program by assuming
that constraints and nodes are in a one-to-one relationship, and we let
$h_i$ be the constraint associated with network node $i$. Here is an
informal description of our first algorithm.
\begin{quote}
  \emph{Constraints Consensus Algorithm}: Beside having access to the
  constraint $h_i$, the $i$th processor state contains a candidate basis
  $\supind{B}{i}$ consisting of $\delta$ elements of $H$.  The processor
  state $\supind{B}{i}$ is initialized to $\delta$ copies of $h_i$. At each
  communication round, the processor performs the following tasks: (i) it
  transmits $\supind{B}{i}$ to its out-neighbors and acquires from its
  in-neighbors their candidate bases; (ii) it solves an abstract optimization
  program with constraint set given by the union of: its constraint $h_i$,
  its candidate basis $\supind{B}{i}$ and its in-neighbors' candidate
  bases; (iii) it updates $\supind{B}{i}$ to be the solution of the
  abstract program computed at step (ii).
\end{quote}
For completeness' sake, the following table presents the algorithm in a way
that is compatible with the model given in
Section~\ref{subsec:networkmodel+distributedalgos}.  The $\subexLP$
algorithm is adopted as local solver for abstract optimization programs.

\bigskip \hrule width \linewidth \smallskip

\noindent\begin{minipage}{0.44\linewidth}\textbf{\texttt{Problem data:}}%
\end{minipage}%
\begin{minipage}{0.56\linewidth}$(\GG,(H,\phi))$%
\end{minipage}

\noindent\begin{minipage}{0.44\linewidth}\textbf{\texttt{Algorithm:}}%
\end{minipage}%
\begin{minipage}{0.56\linewidth}Constraints Consensus%
\end{minipage}

\noindent\begin{minipage}{0.44\linewidth}\textbf{\texttt{Message alphabet:}}%
\end{minipage}%
\begin{minipage}{0.56\linewidth}$\alphabet = H^\delta\union\{\nll\}$%
\end{minipage}

\noindent\begin{minipage}{0.44\linewidth}\textbf{\texttt{Processor state:}}%
\end{minipage}%
\begin{minipage}{0.56\linewidth}$\supind{B}{i}\subset H$ with
  $\card(\supind{B}{i})=\delta$%
\end{minipage}

\noindent\begin{minipage}{0.44\linewidth}\textbf{\texttt{Initialization:}}%
\end{minipage}%
\begin{minipage}{0.56\linewidth}$\supind{B}{i} := \{h_i,\dots,h_i\}$%
\end{minipage}

\bigskip

\noindent\textbf{\texttt{function}}  $\msg(\supind{B}{i}, j)$
\begin{algorithmic}[1]
  \STATE \textbf{return} $\supind{B}{i}$
\end{algorithmic}

\medskip

\noindent\textbf{\texttt{function}}  $\stf(\supind{B}{i}, y)$ \\
\emph{\% executed by node~$i$, with $y_j := \msg(\supind{B}{j},i)=\supind{B}{j}$}\\[-1.9em]
\begin{algorithmic}[1]
  \STATE $\subscr{H}{tmp} := \{h_i\} \union \supind{B}{i} \union \big(
  \union_{j\in \innbrs(i)} y_j \big)$
  \STATE \textbf{return} $\subexLP( \subscr{H}{tmp}, \supind{B}{i})$
\end{algorithmic}

\smallskip \hrule width \linewidth \medskip

\begin{remark}\textbf{\textup{(Constraint re-examination due to lack of persistency)}}
  In order for the algorithm to compute a correct solution, it is necessary
  that each node continuously re-examine its associated constraint
  throughout algorithm execution. In other words, step \algofont{1:} of the
  state-transition function $\stf$ in the algorithm may \emph{not} be
  replaced by $\subscr{H}{tmp} := \supind{B}{i} \union \big( \union_{j\in
    \innbrs(i)} y_j \big)$.  This continuous re-examination is required
  because of the lack of the persistency property discussed after
  Lemma~\ref{lemma:all-is-simple-if-persistent}.  \oprocend
\end{remark}

In the second scenario we consider a time-dependent network with no bounds
on the in-degree of the nodes and on the memory size. In this setting the
execution of the $\subexLP$ may exceed the computation time allocated
between communication rounds.  To deal with this problem, we introduce an
``asynchronous'' version of the network model described in
Section~\ref{sec:network-modeling}: we allow a processor to execute
message-transmission and state-transition functions at instants that are
not necessarily synchronized. Here is an informal description of the
algorithm.
\begin{quote}
  \emph{Multi-round constraints consensus algorithm} Each processor has the
  same message alphabet, processor state, and initialization settings as in
  the previous \emph{constraints consensus algorithm}.
  The processor performs two tasks in parallel. Task \#1: at each
  communication round, the processor transmits to its out-neighbors its
  candidate basis $\supind{B}{i}$ and acquires from its in-neighbors their
  candidate bases.
  Task \#2: independently of communication rounds, the processor repeatedly
  solves an abstract optimization program with constraint set given by the
  union of: its constraint $h_i$, its candidate basis $\supind{B}{i}$ and
  its in-neighbors' candidate bases; the solution of this abstract program
  becomes the new candidate basis $\supind{B}{i}$.
  The abstract program solver is invoked with the most-recently
  available in-neighbors' candidate bases and, throughout its execution,
  this information does not change.
\end{quote}

In the third scenario we consider a time-independent network with no bounds
on the in-degree of the nodes. We suppose that each processor has limited
memory capacity, so that it can store at most $D$ constraints in $H$. The
memory is dimensioned so as to guarantee that the abstract optimization
program is always solvable during two communication rounds (e.g., by
adopting the $\subexLP$ solver). The memory constraint is dealt with by
processing only part of the incoming messages at each round, and by cycling
among incoming messages in such a way as to process all the messages in
multiple rounds.
\begin{quote}
  \emph{Cycling constraints consensus algorithm} The processor state
  contains and initializes a candidate basis $\supind{B}{i}$ as in the
  basic constraints consensus algorithm.  Additionally, the processor state
  includes a counter variable that keeps track of communication rounds.
  At each communication round, the processor performs the following tasks:
  (i) it transmits $\supind{B}{i}$ to its out-neighbors and receives from
  its in-neighbors their candidate bases;
  (ii) among the incoming messages, it chooses to store $D$ messages
  according to a scheduled protocol and the counter variable;
  (iii) it solves an abstract optimization program with constraint set
  given by the union of: its constraint $h_i$, its candidate basis
  $\supind{B}{i}$ and the $D$ candidate bases from its in-neighbors; (iv)
  it updates $\supind{B}{i}$ to be the solution of the abstract program
  computed at step (iii).
\end{quote}

\subsection{Algorithm analysis}
We are now ready to analyze the algorithms. In what follows, we
discuss correctness, halting conditions, memory complexity and time
complexity.

\begin{theorem}\textbf{\textup{(Correctness of the constraints consensus
      algorithm)}}
  \label{thm:correctness-of-constraints-consensus}
  Let $(\GG,(H,\phi))$ be a distributed abstract program with nodes and
  constraints in one-to-one relationship. Assume the time-dependent network
  $\GG$ is jointly strongly connected. Consider a constraint consensus
  algorithm in which a network node initializes its candidate basis to a
  constraint set with finite value.  The following statements hold:
  \begin{enumerate}
  \item\label{fact:monotonicity} along the evolution of the constraints
    consensus algorithm, the basis value $t\mapsto \phi(\supind{B}{i}(t))$
    at each node $i\in\until{n}$ is monotonically non-decreasing and
    converges to a constant finite value in finite time;

  \item\label{fact:convergence} the constraints consensus algorithm solves
    the distributed abstract program $(\GG,(H,\phi))$, that is, in
    finite time the candidate basis $\supind{B}{i}$ at each node $i$ is a
    solution of $(H,\phi)$; and

  \item\label{fact:uniqueness} if the distributed abstract program has a
    unique minimal basis $B_H$, then the final candidate basis
    $\supind{B}{i}$ at each node $i$ is equal to $B_H$.
  \end{enumerate}
\end{theorem}
 \begin{proof}
   From the monotonicity axiom of abstract optimization and from the
   finiteness of $H$, it follows that each sequence
   $\phi(\supind{B}{i}(t))$, $t\in \naturalzero$, is monotone
   non-decreasing, upper bounded and can assume only a finite number of
   values.  Additionally, in finite time, each node has a candidate basis
   that has finite value because, by assumption, one node starts with a
   candidate basis with finite value and the digraph is jointly strongly
   connected.  Therefore, the constraints consensus algorithm at every node
   converges to a constant candidate basis with finite value in a finite
   number of steps. This concludes the proof of
   fact~\ref{fact:monotonicity}.  In what follows, let
   $\supind{B}{1},\dots,\supind{B}{n}$ denote the limiting candidate bases
   at each node in the graph.

   We prove fact~\ref{fact:convergence} in three steps.  First, we proceed
   by contradiction to prove that all the nodes converge to the same value
   (but not necessarily the same basis). The following fact is known: if a
   time-dependent digraph is jointly strongly connected, then the digraph
   contains a time-dependent directed path from any node to any other node
   beginning at any time, that is, for each $t \in \naturalzero$ and each
   pair $(i,j)$, there exists a sequence of nodes $\ell_1,\dots,\ell_k$ and
   a sequence of time instants $t_1, \dots, t_{k+1}\in \naturalzero$ with
   $t \leq t_1 < \dots < t_{k+1}$, such that the directed edges
   $\{(i,\ell_1),(\ell_1,\ell_2), \dots, (\ell_k,j)\}$ belong to the digraph
   at time instants $\{t_1,\dots,t_{k+1}\}$, respectively.  The proof by
   contradiction of a closely related fact is given in
   \cite[Theorem~9.3]{JMH:08}.  Now, suppose that at time $t_0$ all the
   nodes have converged to their limit bases and that there exist at least
   two nodes, say $i$ and $j$, such that $\phi(\supind{B}{i}) \neq
   \phi(\supind{B}{j})$. For $t = t_0 + 1$, for every $k_1 \in
   \outnbrs(i,t_0+1)$, no constraint in $\supind{B}{i}$ violates
   $\supind{B}{k_1}$, otherwise node $k_1$ would compute a new distinct
   basis with strictly larger value, thus violating the assumption that all
   nodes have converged. Therefore, $\phi(\supind{B}{i}) \leq
   \phi(\supind{B}{k_1})$.  Using the same argument at $t = t_0 + 2$, for
   every $k_2 \in \outnbrs(k_1,t_0+2)$, no constraint in $\supind{B}{k_1}$
   violates $\supind{B}{k_2}$.  Therefore, $\phi(\supind{B}{i}) \leq
   \phi(\supind{B}{k_1}) \leq \phi(\supind{B}{k_2})$.  Iterating this
   argument, we can show that for every $S>0$, every node $k$, that is
   reachable from $i$ in the time-dependent digraph with a time-dependent
   directed path of length at most $S$, has a basis $\supind{B}{k}$ such
   that $\phi(\supind{B}{i}) \leq \phi(\supind{B}{k})$.  However, because
   the digraph is jointly strongly connected, we know that there exists a
   time-dependent directed path from node $i$ to node $j$ beginning at time
   $t_0$, thus showing that $\phi(\supind{B}{i}) \leq \phi(\supind{B}{j})$.
   Repeating the same argument by starting from node $j$ we obtain that
   $\phi(\supind{B}{j}) \leq \phi(\supind{B}{i})$. In summary, we showed
   that $\phi(\supind{B}{i}) = \phi(\supind{B}{j})$, thus giving the
   contradiction.  Note that this argument also proves that, if $(i,j)$ is
   an edge of the digraph $\union_{\tau=t}^{+\infty} \GG(\tau)$, then no
   constraint in $i$ violates $\supind{B}{j}$ and, therefore,
   $\phi(\supind{B}{i}\union\supind{B}{j})=\phi(\supind{B}{j})$.  Also, the
   equality
   $\supind{B}{i}\union\supind{B}{j}=\supind{B}{j}\union\supind{B}{i}$
   implies that there exists $\bar{\phi}\in\real$ such that $\bar{\phi}
   =\phi(\supind{B}{i}) =\phi(\supind{B}{j})
   =\phi(\supind{B}{i}\union\supind{B}{j})
   =\phi(\supind{B}{j}\union\supind{B}{i})$ for all $i\in\until{n}$ and
   $(i,j)$ edges of $\union_{\tau=t}^{+\infty} \GG(\tau)$.

   Second, we claim that the value of the basis at each node is equal to
   the value of the union of all the bases. In other words, we claim that
   \begin{equation}
     \label{eq:globalopt}
     \bar{\phi} = \phi(\supind{B}{1}\union\cdots \union \supind{B}{n}).
   \end{equation}
   We prove equation~\eqref{eq:globalopt} by induction.  First, we note
   that $\bar{\phi} = \phi(\supind{B}{i}\union\supind{B}{j})$ for any nodes
   $i$ and $j$ such that either $(i,j)$ or $(j,i)$ is a directed edge in
   $\union_{\tau=t}^{+\infty} \GG(\tau)$.  Without loss of generality, let
   us assume $i=1$ and $j=2$.  Now assume that
   \begin{equation}
     \label{eq:no-idea}
     \phi(\supind{B}{1}\union\cdots \union \supind{B}{k})
     = \bar{\phi},
   \end{equation}
   for an arbitrary $k$-dimensional weakly-connected subgraph $G_k$ of
   $\union_{\tau=t}^{+\infty} \GG(\tau)$ and we prove such a statement for
   a weakly-connected subgraph of dimension $k+1$ containing $G_k$.  By
   contradiction, we assume the statement is not true for $k+1$.  Assuming,
   without loss of generality, that node $k+1$ is connected to $G_k$ in
   $\union_{\tau=t}^{+\infty} \GG(\tau)$, we aim to find a contradiction
   with the statement
   \begin{equation}
     \label{eq:absurb}
     \phi(\supind{B}{1}\union\cdots \union \supind{B}{k}
     \union \supind{B}{k+1} )
     > \bar{\phi}.
   \end{equation}
   Plugging the induction assumption into equation~\eqref{eq:absurb}, we
   have
   \begin{equation}
     \label{eq:n+1greater-n}
     \phi(\supind{B}{1}\union\cdots \union \supind{B}{k} \union
     \supind{B}{k+1} ) >
     \phi(\supind{B}{1}\union\cdots \union \supind{B}{k}).
   \end{equation}
   From Lemma~\ref{lemma:locality2} with $F=\supind{B}{1}\union\cdots
   \union \supind{B}{k}$ and $G=\supind{B}{k+1}$ and noting
   equation~\eqref{eq:n+1greater-n}, we conclude that there exists
   $g\in\supind{B}{k+1}$ such that
   \begin{equation}
     \label{eq:three}
     \phi(\supind{B}{1}\union\cdots \union
     \supind{B}{k}\union\{g\}) > \phi(\supind{B}{1}\union\cdots \union
     \supind{B}{k}).
   \end{equation}
   Next, select a node $p\in\until{k}$ such that either $(p,k+1)$ or
   $(k+1,p)$ is a directed edge of $\union_{\tau=t}^{+\infty} \GG(\tau)$
   and note that $\supind{B}{p} \subset \supind{B}{1}\union\cdots \union
   \supind{B}{k}$ and $\phi(\supind{B}{p}) = \phi(
   \supind{B}{1}\union\cdots \union \supind{B}{k})$ by the induction
   assumption. From these two facts together with
   equation~\eqref{eq:three}, the locality property implies that
   \begin{equation}
     \label{eq:four}
     \phi(\supind{B}{p}\union\{g\}) >  \phi(\supind{B}{p}).
   \end{equation}
   Finally, the contradiction follows by noting:
   \begin{align*}
     \label{eq:six}
     \phi(\supind{B}{p}) &\overset{\text{bases have converged}}{=}
     \phi(\supind{B}{p}\union \supind{B}{k+1})
     \overset{\text{monotonicity}}{\geq}
     \phi(\supind{B}{p}\union \{g\})
     \overset{\text{by equation}~\eqref{eq:four}}{>} \phi(\supind{B}{p}).
   \end{align*}
   This concludes the proof of equation~\eqref{eq:globalopt}.

   Third and final, because no constraint in $\{h_1,\ldots,h_n\}$ violates
   the set $\supind{B}{1}\union\cdots\union\supind{B}{n}$ and because
   $\supind{B}{1}\union\cdots\union\supind{B}{n} \subset H$,
   Lemma~\ref{lemma:locality2} and equation~\eqref{eq:globalopt} together
   imply
   \begin{equation*}
     \bar{\phi} = \phi(\supind{B}{1}\union\cdots\union\supind{B}{n})
     =\phi(H).
   \end{equation*}
   This equality proves that in a finite number of rounds the candidate
   basis at each node is a solution to $(H,\phi)$ and, therefore, this
   concludes our proof of fact~\ref{fact:convergence}.  The proof of
   fact~\ref{fact:uniqueness} is straightforward.
 \end{proof}

\begin{remark}\textbf{\textup{(Correctness of multi-round and cycling constraints consensus)}}
  Correctness of the other two versions of the constraints consensus
  algorithm may be established along the same lines. For example, it is
  immediate to establish that the basis at each node reaches a constant
  value in finite time.  It is easy to show that this constant value is the
  solution of the abstract optimization program for the multi-round algorithm
  over a time-dependent graph.  For the cycling algorithm over a
  time-independent graph we note that the procedure used to process the
  incoming data is equivalent to considering a time-dependent graph whose
  edges change with that law.  \oprocend
\end{remark}

\begin{theorem}[Halting condition]
  \label{thm:halting-condition}
  Consider a network described by a time-independent strongly-connected
  digraph $\GG$ implementing a constraints consensus algorithm in which a
  network node initializes its candidate basis to a constraint set with
  finite value. Each processor can halt the algorithm execution if the
  value of its basis has not changed after $2\diam(\GG)+1$ communication
  rounds.
\end{theorem}
\begin{proof}
  For all $t\in\naturalzero$ and for every $(i,j)$ edge of $\GG$,
  \begin{equation}
    \label{eq:monotonicity-(i,j)-t_t+1}
    \phi(\supind{B}{i}(t)) \leq \phi(\supind{B}{j}(t+1)),
  \end{equation}
  because, by construction along the constraints consensus algorithm, the
  basis $\supind{B}{j}(t+1)$ is not violated by any constraint in the basis
  $\supind{B}{i}(t)$.  Assume that node $i$ satisfies $\supind{B}{i}(t) =
  B$ for all $t \in \fromto{t_0}{t_0+2\diam(\GG)}$, and pick any other node
  $j$. Without loss of generality, set $t_0=0$. Because of
  equation~\eqref{eq:monotonicity-(i,j)-t_t+1}, if $k_1\in\outnbrs(i)$,
  then $\phi(\supind{B}{k_1}(1)) \geq \phi(B)$ and, recursively, if
  $k_2\in\outnbrs(k_1)$, then $\phi(\supind{B}{k_2}(2)) \geq
  \phi(\supind{B}{k_1}(1)) \geq \phi(B)$.  Therefore, iterating this
  argument $\dist(i,j)$ times, the node $j$ satisfies
  $\phi(\supind{B}{j}(\dist(i,j))) \geq \phi(B)$.  Now, consider the
  out-neighbors of node $j$.  For every $k_3\in\outnbrs(j)$, it must hold
  that $\phi(\supind{B}{k_3}(\dist(i,j)+1)) \geq
  \phi(\supind{B}{j}(\dist(i,j)))$.  Iterating this argument $\dist(j,i)$
  times, the node $i$ satisfies $\phi(\supind{B}{i}(\dist(i,j)+\dist(j,i)))
  \geq \phi(\supind{B}{j}(\dist(i,j)))$.  In summary, because
  $\dist(i,j)+\dist(j,i)\leq 2\diam(\GG)$, we know that
  $\phi(\supind{B}{i}(\dist(i,j)+\dist(j,i)))=\phi(B)$ and, in turn, that
  \begin{equation*}
    \phi(B) \geq \phi(\supind{B}{j}(\dist(i,j))) \geq \phi(B).
  \end{equation*}
  Thus, if basis $i$ does not change for $2\diam(\GG)+1$ time instants,
  then its value will never change afterwards because all bases
  $\supind{B}{j}$, for $j\in\until{n}$, have cost equal to $\phi(B)$ at
  least as early as time equal to $\diam(\GG)+1$. Therefore, node $i$ has
  sufficient information to stop the algorithm after a $2\diam(\GG)+1$
  duration without value improvement.
\end{proof}

Next, let us state some simple memory complexity bounds for the three
algorithms.  Assume that $\delta$ is the combinatorial dimension of the
abstract program $(H, \phi)$ and call a memory unit is the amount of
memory required to store a constraint in $H$.  Each node $i$ of the
network requires $1+\delta(1+\card(\innbrs(i)))\in O(n)$ memory units in
order to implement the constraints consensus algorithm and its multi-round
variation, and $1+\delta(1+D) \in O(1)$ in order to implement the cycling
constraints consensus algorithm.

We conclude this section with some incomplete results about the
\emph{completion time} of the constraints consensus algorithm, i.e., the
number of communication rounds required for a solution of the distributed
abstract program, and about the \emph{time complexity}, i.e., the
functional dependence of the completion time on the number of
agents. First, it is straightforward to show that there exist distributed
abstract programs of dimension $n$ for which the time complexity can
be lower bounded by $\Omega(n)$. Indeed, it takes order $n$ communication
rounds to propagate information across a path graph of order $n$.  On the
other hand, it is also easy to provide a \emph{loose} upper bound by noting
that (i) the number of possible distinct bases for an abstract optimization
program with $n$ constraints and combinatorial dimension $\delta$ is upper
bounded by $n^\delta$, and (ii) at each communication round at least one
node in the network increases its basis. Therefore, the worst-case time
complexity is upper bounded by $n^{\delta+1}$.  It is our conjecture that
the average time complexity of the constraints consensus algorithm is much
better than the loose analysis we are able to provide so far.

\begin{conjecture}[Linear average time complexity]
  \label{conjecture:linear}
  Over the set of time-independent strongly-connected digraphs and distributed
  abstract programs, the average time complexity of the constraints
  consensus algorithm belongs to $O(\diam(\GG))$. \oprocend
\end{conjecture}

\section{Monte Carlo analysis of the time complexity of constraints
  consensus}
\label{sec:computations}
This section presents a simulation-based analysis of the time complexity of
the constraints consensus algorithm for stochastically generated
distributed abstract programs.  To define a numerical experiment, i.e., a
stochastically-generated distributed abstract program, we need to specify
(1) the communication graph, (2) the abstract optimization problem and (3)
various parameters describing a nominal set of problems and some variations
of interest.  We discuss these three degrees of freedom in the next three
subsections and perform two sets of Monte Carlo analysis in the two
subsequent subsections.

\subsection{Communication graph models}
We consider time-independent undirected communication graphs generated
according to one of the following three graph models. The first model is
the \emph{line graph}. It has bounded node degree and the largest
diameter. Then we consider two random graphs models, namely the well-known
\emph{Erd\H{o}s-R\`enyi graph} and \emph{random geometric graph}. In the
Erd\H{o}s-R\`enyi graph an edge is set between each pair of nodes with
equal probability $p$ independently of the other edges. A known
result~\cite{RA-ALB:02} is that the average degree of the nodes is $p
n$. Also, if $p = (1 + \epsilon) \log(n) / n$, $\epsilon > 0$, the graph is
almost surely connected and the average diameter of the graph is $\log(n) /
\log(pn)$. Indeed, we use the probability $(1 + \epsilon) \log(n) / n$ to
generate the graph. This gives an average node degree that is unbounded as
$n$ grows, but the growth is logarithmic and so the local computations are
still tractable. A random geometric graph in a bounded region is generated
by (i) placing nodes at locations that are drawn at random uniformly and
independently on the region and (ii) connecting two vertices if and only if
the distance between them is less than or equal to a threshold $r>0$. We
generate random geometric graphs in a unit-length square of $\real^2$.  To
obtain a connected graph we set the radius $r$ to the minimum value that
guarantees connectivity.

\subsection{Linear programs models}
In our experiments we consider stochastically-generated linear programs
according to well-known models proposed in the literature. A detailed
survey on stochastic models for linear programs and their use to study the
performance of the simplex method is \cite{RS:87}.  We consider standard
LPs in $d$-dimensions with $n$ constraints of the form
\begin{equation*}
  \begin{split}
    \min  &\quad c^T x\\
    \subject &\quad A x \leq b,
  \end{split}
\end{equation*}
where $A\in\real^{n\times d}$, $b\in\real^n$ and $c\in\real^d$ are
generated according to the following stochastic models.

\emph{Model A.} In this model the elements $A_{ij}$ and $c_j$ are
independently drawn from the standard Gaussian distribution. The vector $b$
is defined by $b_i = \big(\sum_{j=1}^d A^2_{ij}\big)^{1/2}$,
$i\in\until{n}$.  This corresponds to generating hyperplanes (corresponding
to the constraints) whose normal vectors are uniformly distributed on the
unit sphere and that are at unit distance from the origin. The LP problems
generated according to this model are always feasible. This model was
originally proposed by~\cite{JRD-DGK-JWT:77} and is a special case of a
class of models introduced by~\cite{KHB:77} see \cite{RS:87} for
details.\footnote{Our Model A is the one indicated as Model N
  in~\cite{RS:87}.  A slightly different model, called Model O, features
  hyperplanes at random distance from the origin. Numerical experiments
  show that, for large $n/d$, Model N problems require more iterations on
  average than Model O ones. Indeed, for fixed $d$ and increasing $n$, all
  constraints in Model N problems are relevant, whereas many constraints in
  Model O problems are not.}

\emph{Model B.} In this model the vector $c$ is obtained as $c = A^T
\hat{c}$. The vector $(b, \hat{c}) \in \real^{2n}$ is uniformly randomly
generated in $[0,1]^{2n}$ and $A$ is a standard Gaussian random matrix
independent of $(b, \hat{c})$.  The LP problems generated according to this
model are always feasible. This LP model, with a more general stochastic
model for $(b, \hat{c})$, was proposed by Todd in \cite{MJT:91} (where it
is the model indicated as ``Model 1'' in a collection of three).

\subsection{Nominal problems and variations}
First, as \emph{nominal set of problems} we consider a set of distributed
abstract programs with the following characteristics: $d=4$,
$n\in\{20,\dots,240\}$, the graphs are equal to the line graphs of
dimension $n$, and the linear problems are generated from Model A.

Second, as variations of the nominal set of problems, we generate LPs of
dimension $d\in\{2,3,4,5\}$ with a number of constraints
$n\in\{20,40,\dots,240\}$. For each value of $d$ we generate a graph
according to one of the three graph models and an LP according to one of
the two LP models. For each configuration (dimension, number of
constraints, graph model, and LP model), we generate different problems, we
solve each problem with the constraints consensus algorithm, and we store
worst-case and average completion time.

Results for the nominal set of problems and for its variations are given in
the next sections.

\subsection{Time complexity results for nominal problems}
For the nominal set of problems, we study the time complexity via the
Student $t$-test and via Monte Carlo probability estimation.

For each value of $n$, we perform a Student's $t$-test with the null
hypothesis being that the average completion time divided by the graph
diameter is greater than $1.5$ -- against the alternate hypothesis that the
same ratio is less than or equal to that (at the $95\%$ confidence
level). (Note that the diameter is $(n-1)$.) The results for
$n\in\{200,220,240\}$ are shown in Table~\ref{tab:ttest}. The tests show
that we can reject the null hypothesis.
In Figure~\ref{fig:Nst_vs_n_line} we show the linear dependence of the
completion time with respect to the number of agents (and therefore with
respect to the diameter) and provide the corresponding $95\%$ confidence
intervals.

\begin{table}[h]
  \begin{center}
    \setlength{\tabcolsep}{4pt}
    \begin{tabular}[c]{cccccc}
      \parbox{.16\linewidth}{number of constraints} &
      \parbox{.16\linewidth}{average completion time/diam} &
      \parbox{.15\linewidth}{standard deviation} &
      df &
      $t$-value &
      $p$-value \\[.75em] \hline
      \\[-.75em]
      $240$ & $1.21$ & $0.36$ & $99$ & $-7.73$ & $4.3 \times 10^{-12}$\\
      $220$ & $1.16$ & $0.31$ & $99$ & $-10.90$ & $6.0 \times 10^{-19}$\\
      $200$ & $1.27$ & $0.36$ & $99$ & $-6.49$ & $1.7 \times 10^{-9}$
    \end{tabular}
    \caption{Student's t-test results for demonstrating the linear
      dependence of the completion time with respect to the
      diameter. Problem: Graph = line graph, LP model = Model A, $d=4$,
      $n\in\{200,220,240\}$, $\subscr{N}{run}=100$, null hypothesis:
      average completion time $/ n > 1.5$.}
\label{tab:ttest}
\end{center}
\end{table}

\begin{figure}[h]
\begin{center}
  \includegraphics[width=.4\linewidth]{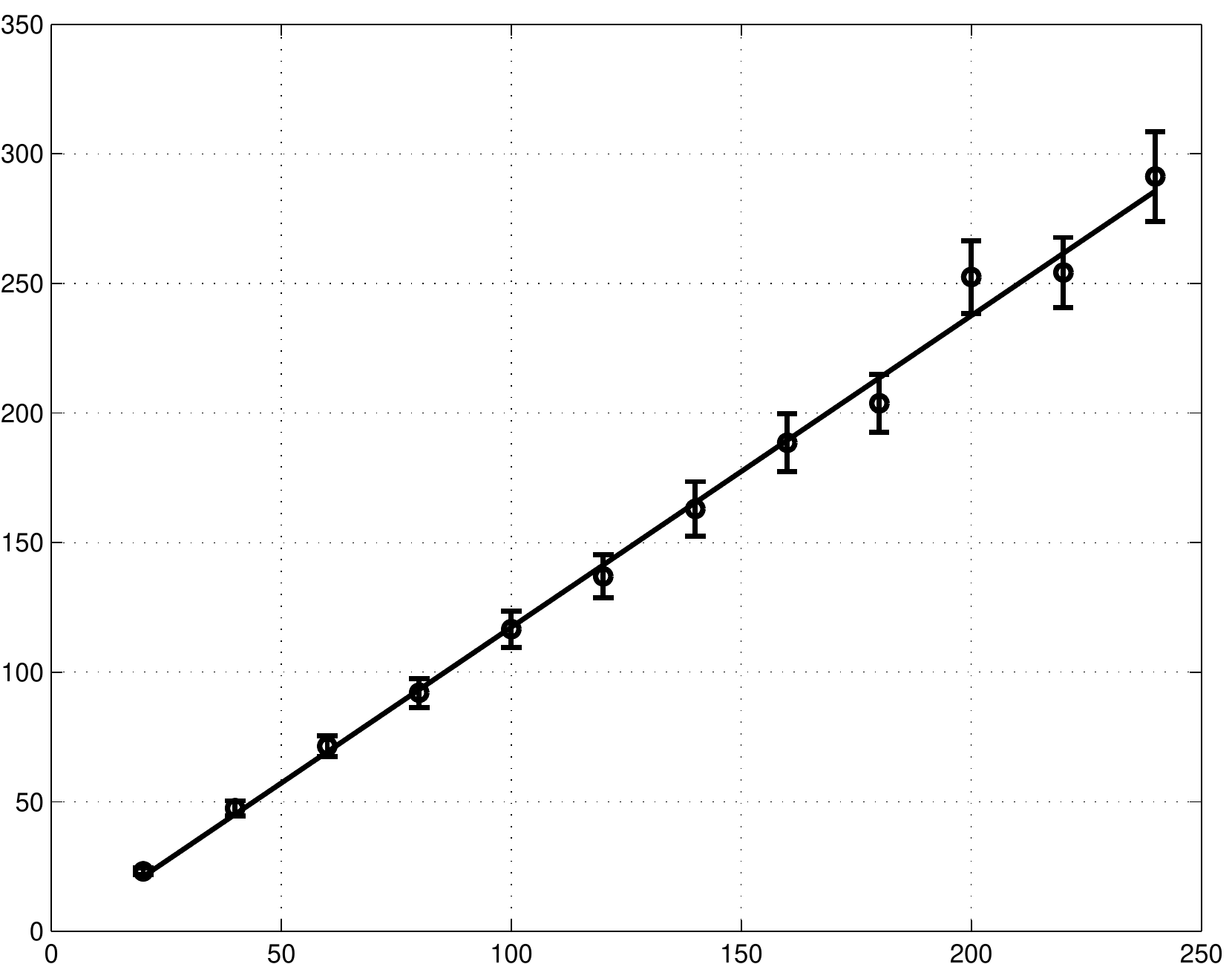}%
  \caption{Average completion time for increasing number of constraints
    $n$. Specifically: Graph = line graph, LP model = Model A, $d=4$,
    $n\in\{20,\dots,240\}$, $\subscr{N}{run}=100$, performance =
    average. The solid line is the least-square interpolation of the
    average completion times.}
  \label{fig:Nst_vs_n_line}
\end{center}
\end{figure}

Next, we aim to upper bound the worst-case completion time. To do so we use
a Monte Carlo probability estimation method, that we review
from~\cite{RT-GC-FD:05}.

\begin{remark}[Probability estimation via Monte
  Carlo]\label{rem:probability-estimation-via-montecarlo}
  We aim to estimate the probability that a random variable is less than or
  equal to a given threshold. Let $Q$ be a compact set and $\Delta$ be a
  random variable taking values in $Q$. Given a scalar threshold $\gamma$,
  define the probability $p(\gamma) = \text{Pr}\{J(\Delta)\leq \gamma\}$,
  where $\map{J}{Q}{\real}$ is a given measurable performance function. We
  estimate $p(\gamma)$ as follows. First, we generate $N\in\natural$
  independent identically distributed random samples $\Delta^{(1)}, \ldots,
  \Delta^{(N)}$. Second, we define the indicator function
  $\map{\mathbb{I}_{J,\gamma}}{Q}{\{0,1\}}$ by
  $\mathbb{I}_{J,\gamma}(\Delta) = 1$ if $J(\Delta)\leq\gamma$, and $0$
  otherwise. Third and final, we compute the \emph{empirical probability}
  as
  \begin{equation*}
    \hat{p}_N(\gamma) = \frac{1}{N}\sum_{i=1}^N \mathbb{I}_{J,\gamma}(\Delta^{(i)}).
  \end{equation*}
  Next, we adopt the Chernoff bound in order to provide a bound on the
  number of random samples required for a certain level of accuracy on the
  probability estimate.  For any accuracy $\epsilon \in (0, 1)$ and
  confidence level $1-\eta \in (0, 1)$, we know that $\|\hat{p}_N(\gamma)
  - p(\gamma)\| < \epsilon$ with probability greater than $1-\eta$ if
  \begin{equation}
    \label{bound:Chernoff}
    N \geq \frac{1}{2\epsilon^2}\log\frac{2}{\eta}.
  \end{equation}
  For $\epsilon = \eta = 0.01$, the Chernoff bound~\eqref{bound:Chernoff}
  is satisfied by $N = 27000$ samples.  \oprocend
\end{remark}

Adopting the same notation as in
Remark~\ref{rem:probability-estimation-via-montecarlo}, here is our setup:
First, the random variable of interest is a collection of $n$ unit-length
vectors (i.e., our random variable takes values in the compact space
$\setdef{x\in\real^4}{\|x\|=1}^n$).  Second, the function $J$ is the
completion time of the constraints consensus algorithm in solving a nominal
problem with constraints determined by $\Delta$.  Third, we want to
estimate the probability that, for $n\in\{40,60,80\}$, the completion time
is less than or equal to $4$ times the diameter of the chain graph of
dimension $n$.  In order to achieve an accuracy $0.01$ with confidence
level $99\%$, we run $\subscr{N}{run}=27000$ experiments for each value of
$n$ and we compute the maximum completion time in each case.  The
experiments show that for each $n$ the worst-case completion time is less
than $3.4$ times the graph diameter.  Therefore, we have established the
following statement.
\begin{quote}
  With $99\%$ confidence level, there is at least $99\%$ probability that a
  nominal problem ($d=4$, graph = line graph, LP model = Model A) with
  number of constraints $n\in\{40,60,80\}$ is solved by the constraints
  consensus algorithm in time upper bounded by $4(n-1)$.
\end{quote}

\subsection{Time complexity results for variations of the nominal problems}
Next we perform a comparison among different graph models, LP models and LP
dimensions. In order to compare the performance of different graphs we
consider problems with: Graph $\in$ \{line graph, Erd\H{o}s-R\`enyi graph,
random geometric graph\}, LP model = Model A, $d=4$,
$n\in\{20,\dots,240\}$, $\subscr{N}{run}=100$. We compute the average
completion time to diameter ratio for increasing values of $n$. The results
with the $95\%$ confidence interval are shown in
Figure~\ref{fig:Nst_vs_n_graph}.

\begin{figure}[h]
\begin{center}
  \includegraphics[width=.4\linewidth]{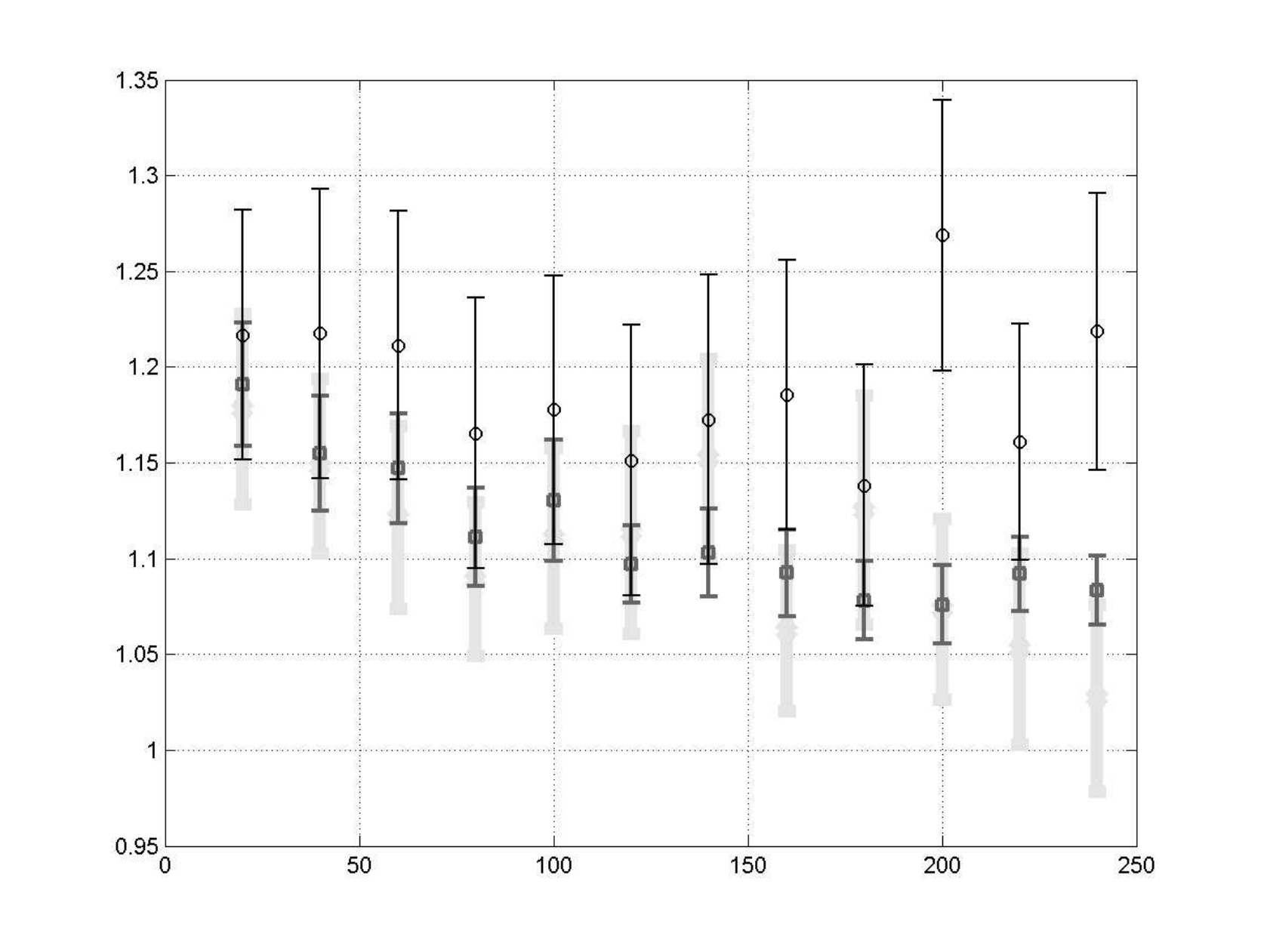}%
  \caption{Comparing completion time to diameter ratio for distinct graphs
    and for increasing number of constraints $n$. Specifically: Graph $\in$
    \{line graph (circle), Erd\H{o}s-R\`enyi graph (square), random
    geometric graph (diamond)\}, LP model = Model A, $d=4$,
    $n\in\{20,\dots,240\}$, $\subscr{N}{run}=100$, performance = average.}.
  \label{fig:Nst_vs_n_graph}
\end{center}
\end{figure}

To compare the performance for different
LP models, we consider problems with: Graph = line graph, LP model $\in$
\{Model A, Model B\}, $d=4$, $n\in\{20,\dots,240\}$,
$\subscr{N}{run}=100$. The results with the $95\%$ confidence interval are
shown in Figure~\ref{fig:Nst_vs_n_line_LP}.

\begin{figure}[h]
\begin{center}
  \includegraphics[width=.4\linewidth]{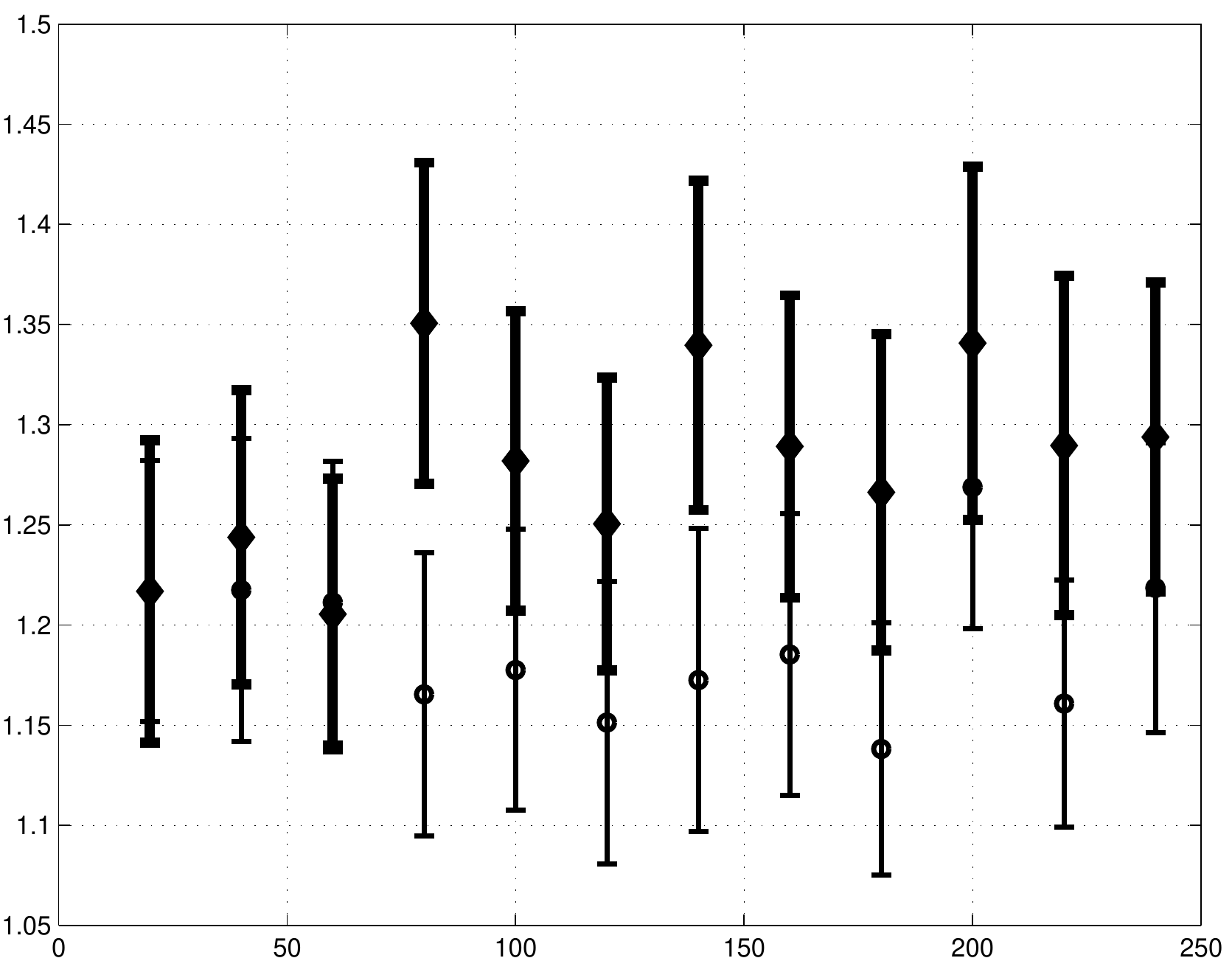}%
  \caption{Comparing completion time to diameter ratio for distinct LP
    models and for increasing number of constraints $n$. Specifically:
    Graph = line graph, LP model $\in$ \{Model A (circle), Model B
    (diamond)\} $d = 4$, $n\in\{20,\dots,240\}$, $\subscr{N}{run}=100$,
    performance = average.}.
  \label{fig:Nst_vs_n_line_LP}
\end{center}
\end{figure}

Next, to compare the performance for different dimensions $d$, we consider
problems with: Graph = line graph, LP model = Model A, $d\in\{2,3,4,5\}$,
$n\in\{20,\dots,240\}$, $\subscr{N}{run}=100$. The results with the $95\%$
confidence interval are shown in Figure~\ref{fig:Nst_vs_n_line_dim}.
The comparisons show that, the linear dependence of the completion time
with respect to the number of constraints is not affected by the graph
topology, the LP model and the dimension $d$. As regards the dimension $d$,
as expected, for fixed $n$ the average completion time grows with the
dimension. Also, the growth appears to be linear for $d\geq 3$ (for $d=2$
the algorithm seems to perform much better). In
Figure~\ref{fig:NstDiam_vs_d_line} we plot the least square value of the
completion time to diameter ratio over the number of agents $n$ versus the
dimension $d$.

\begin{figure}[h]
\begin{center}
  \includegraphics[width=.4\linewidth]{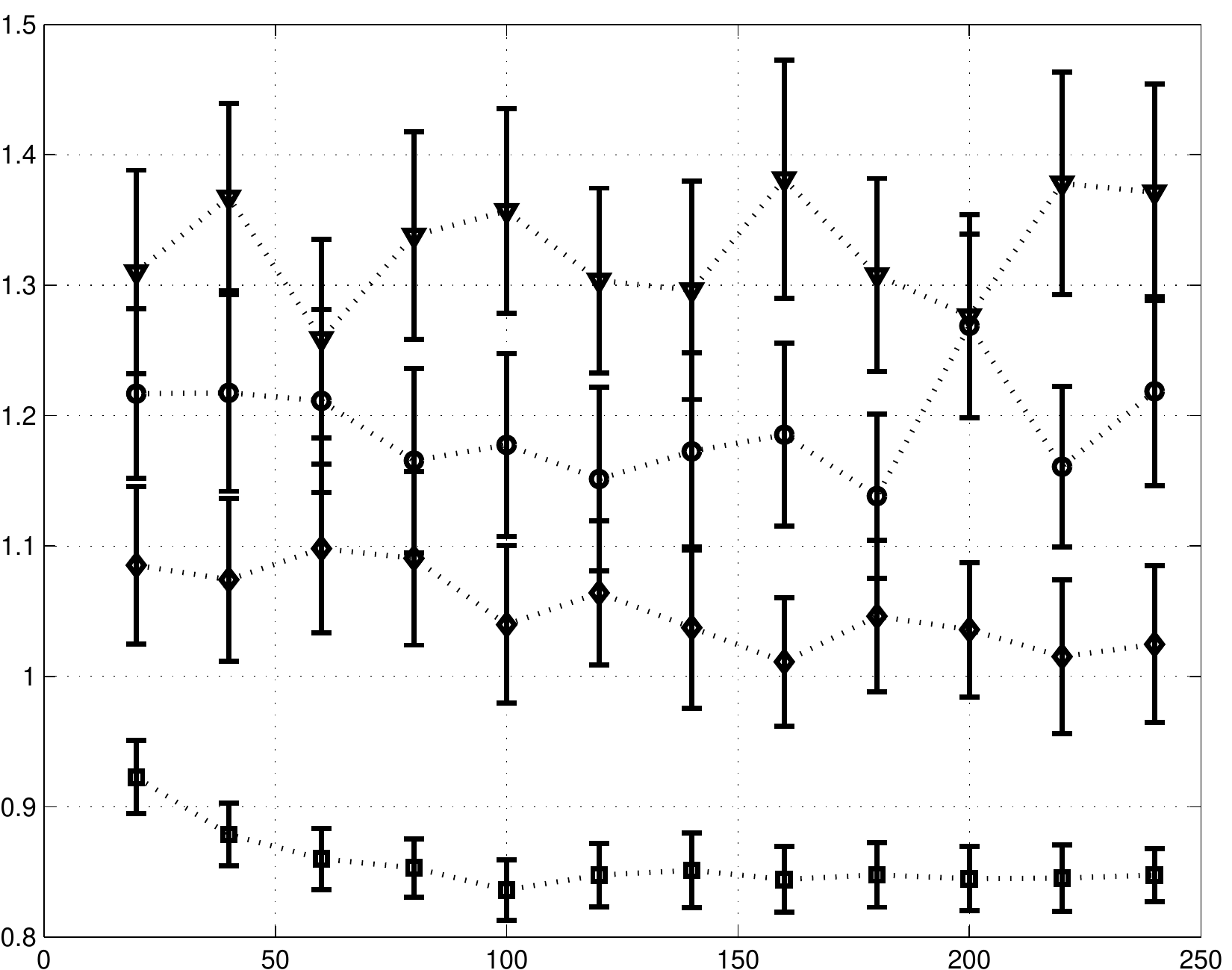}%
  \caption{Comparing completion time to diameter ratio for distinct problem
    dimensions $d$ and for increasing number of constraints
    $n$. Specifically: Graph = line graph, LP model = Model A, $d\in\{2
    (square),3 (diamond),4 (circle),5 (triangle)\}$,
    $n\in\{20,\dots,240\}$, $\subscr{N}{run}=100$, performance = average.}.
  \label{fig:Nst_vs_n_line_dim}
\end{center}
\end{figure}

\begin{figure}[h]
\begin{center}
  \includegraphics[width=.4\linewidth]{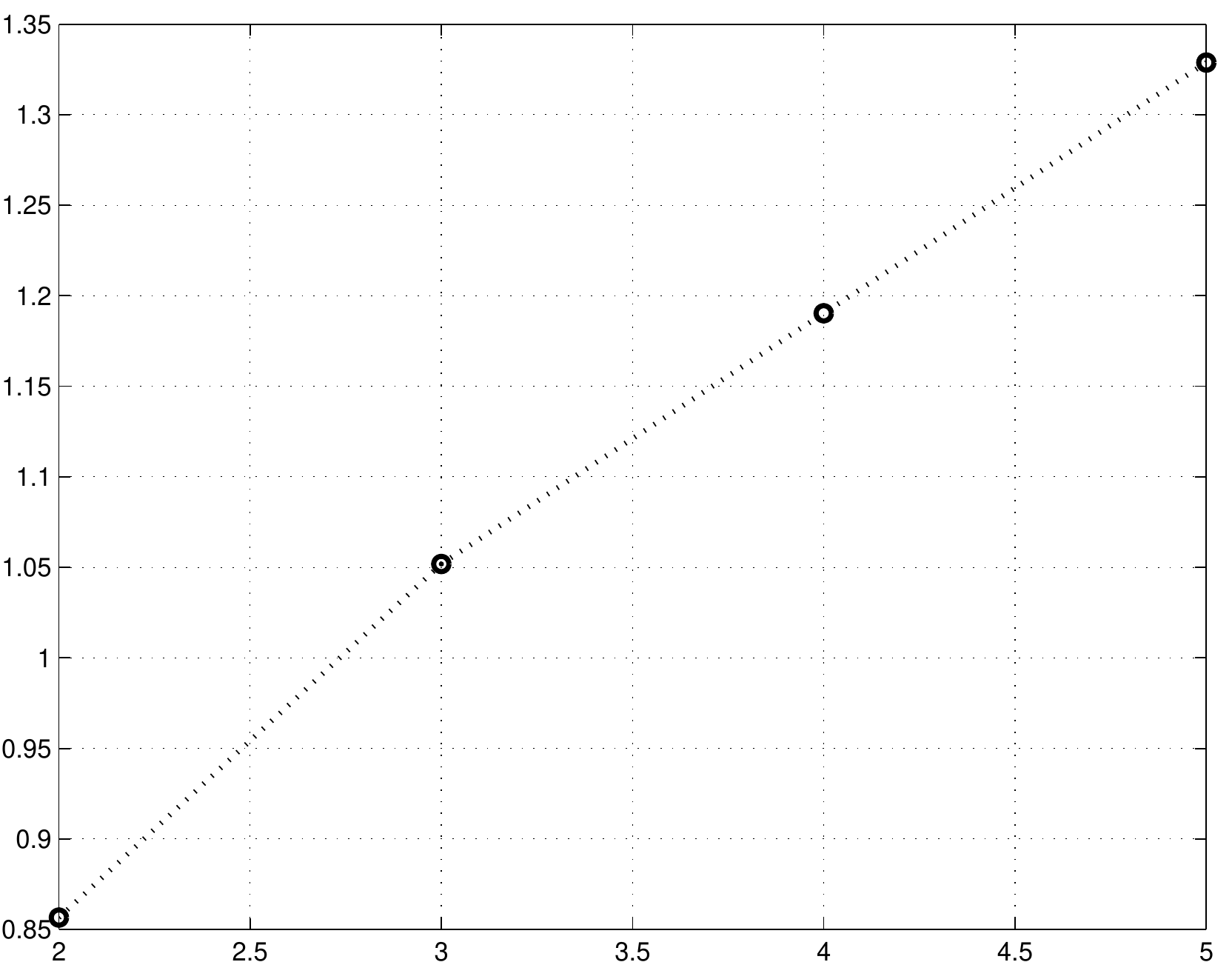}%
  \caption{Comparing completion time to diameter ratio (least square value
    in $n$) for increasing problem dimensions $d$. Specifically: Graph =
    line graph, LP model = Model A, $d\in\{2,3,4,5\}$,
    $n\in\{20,\dots,240\}$, $\subscr{N}{run}=100$, performance = average.}.
  \label{fig:NstDiam_vs_d_line}
\end{center}
\end{figure}

\section{Application to target localization in sensor networks}
\label{sec:target-localization}
In this section we discuss an application of distributed abstract
programming to sensor networks, namely a distributed solution for target
localization.  Essentially, we propose a distributed algorithm to
approximately compute the intersection of time-varying convex polytopes.

\subsection{Motion and sensing models}
We consider a target moving on the plane with unknown but bounded
velocity. The first-order dynamics is given by
\begin{equation}  \label{eq:target_dyn}
  p(t + 1) = p(t) + v(t),
\end{equation}
where $p(t)\in{[\xmin,\xmax]\times[\ymin,\ymax]}\subset\real^2$ is the
target position at time $t\in\naturalzero$ and $v(t)\in\real^2$ is unknown
but satisfies $\norm{v}{}\leq \subscr{v}{max}$ for given $\subscr{v}{max}$.

A set of sensors $\until{n}$ deployed in the environment measures the
target position. We assume that the measurement noise is unknown and such
that, at each time instant, each sensor $i$ measures a possibly-unbounded
region of the plane, $\supind{h}{i}(p(t)) \subset \real^2$, containing the
target. The set $M(p(t)) =
{\ensuremath{\operatorname{\cap}}}_{i\in\until{n}} \supind{h}{i}(p(t))$,
called the measurement set, provides the best estimate of the target
position based on instantaneous measures only.  An example scenario is
illustrated in Figure~\ref{fig:target-localization-scen}.  In what follows,
we make a critical assumption and a convenient one. First, we assume that
each measured region $\supind{h}{i}(p(t))$ is a possibly-unbounded convex
polygon. Second, for simplicity of notation, we assume that each measured
region $\supind{h}{i}(p(t))$ is a half-plane, so that the measurement set
$M(p(t))$ is a non-empty possibly-unbounded convex polygon equal with up to
$n$ edges.
\begin{figure}[h]
  \begin{center}
    \includegraphics[width=.36\linewidth]{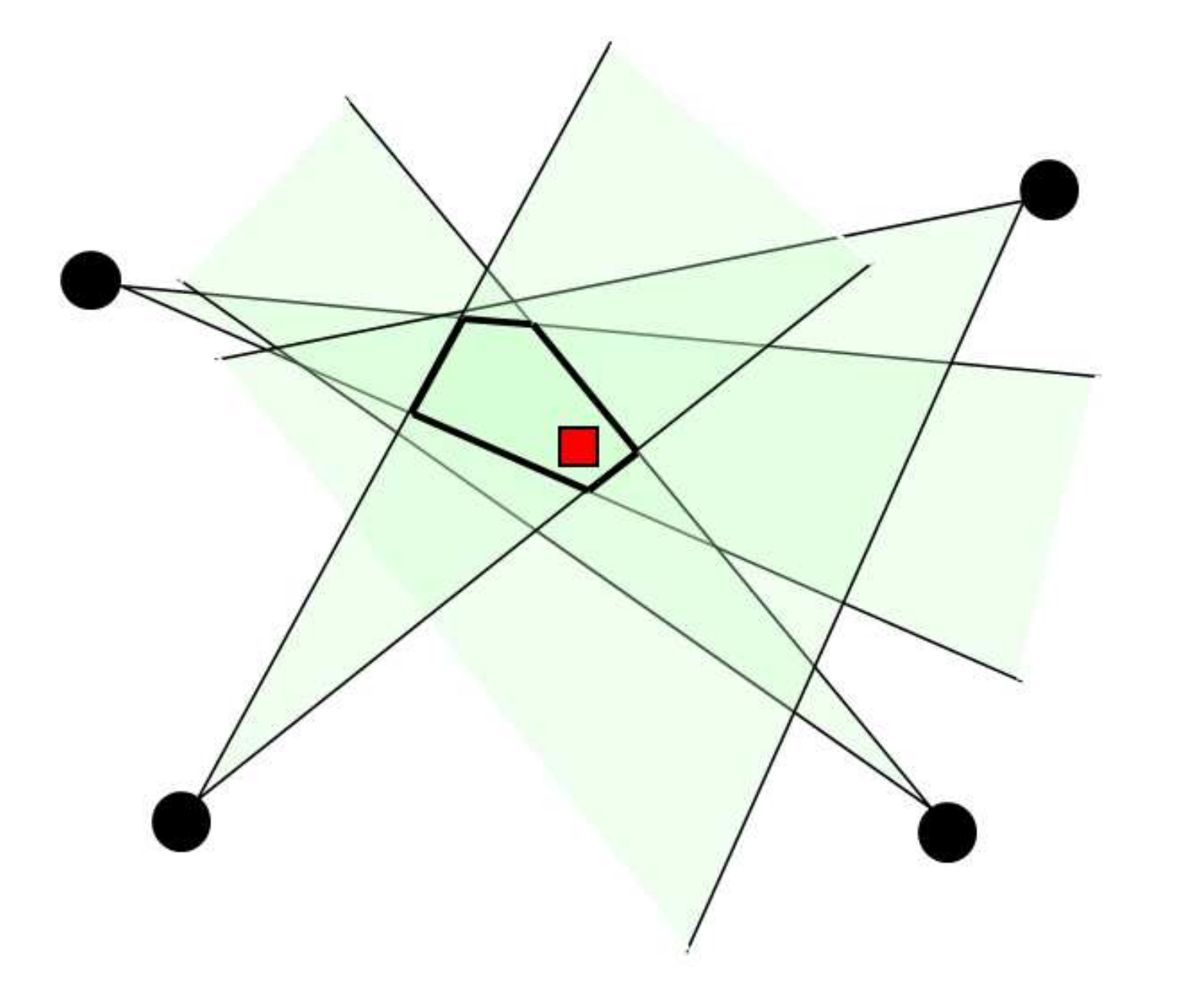}%
    \caption{The measurement set is the intersection of the $4$ sensor
      measurements.}
    \label{fig:target-localization-scen}
  \end{center}
\end{figure}

\subsection{Set-membership localization of a moving target}

\begin{problem}[Set-membership localization]
  Compute the smallest set $E(t)\subset\real^2$ that contains the target
  position $p(t)$ at $t\in\integernonnegative$ and that is consistent with
  the dynamic model~\eqref{eq:target_dyn} and the sensor measurements
  $\supind{h}{i}(p(t))$, $i\in\until{n}$. \oprocend
\end{problem}

We adopt the \emph{set-membership approach} described
in~\cite{AG-AV:01}.\footnote{In~\cite{AG-AV:01} a moving vehicle localizes
  itself by measuring landmarks at known positions. This
  ``self-localization'' scenario leads to a mathematical model closely
  related to the one we consider.  }  For $t\in\natural$, define the sets
$E(t|t-1)$ and $E(t|t)$ as the \emph{feasible position sets} containing all
the target positions at time $t$ that are compatible with the dynamics and
the available measurements up to time $t-1$ and $t$, respectively.  With
this notation, the recursion equations are:
\begin{subequations}
  \label{eq:set_memb_recurs}
  \begin{align}
    E(0|0)    &= M(p(0)),
    \label{eq:set_memb_recurs_initialization} \\
    E(t|t-1)  &=  E(t-1|t-1) + \cball{\subscr{v}{max}}{0},
    \label{eq:set_memb_recurs_time-update}
    \\
    E(t|t) &= E(t|t-1) \intersection M(p(t)) ,
    \label{eq:set_memb_recurs_measurement-update}
  \end{align}
\end{subequations}
where the sum set of two sets $A+B$ is defined as $\setdef{a+b}{a\in A
  \text{ and }b\in B}$.  Equation~\eqref{eq:set_memb_recurs_time-update} is
justified as follows: since the target speed satisfies
$\norm{v}{}\leq\subscr{v}{max}$, if the target is at position $p$ at time
$t$, then the target must be inside $\cball{\subscr{v}{max}\tau}{p}$ at
time $t+\tau$, for any positive $\tau$.
Equation~\eqref{eq:set_memb_recurs_measurement-update} is a direct
consequence of the definition of measurement set.
Equations~\eqref{eq:set_memb_recurs_initialization},
~\eqref{eq:set_memb_recurs_time-update},
and~\eqref{eq:set_memb_recurs_measurement-update} are referred to as
\emph{initialization}, \emph{time update} and \emph{measurement update},
respectively. The time and measurement updates are akin to prediction and
correction steps in Kalman filtering.

\begin{figure}[h]
\begin{center}
  \includegraphics[width=.65\linewidth]{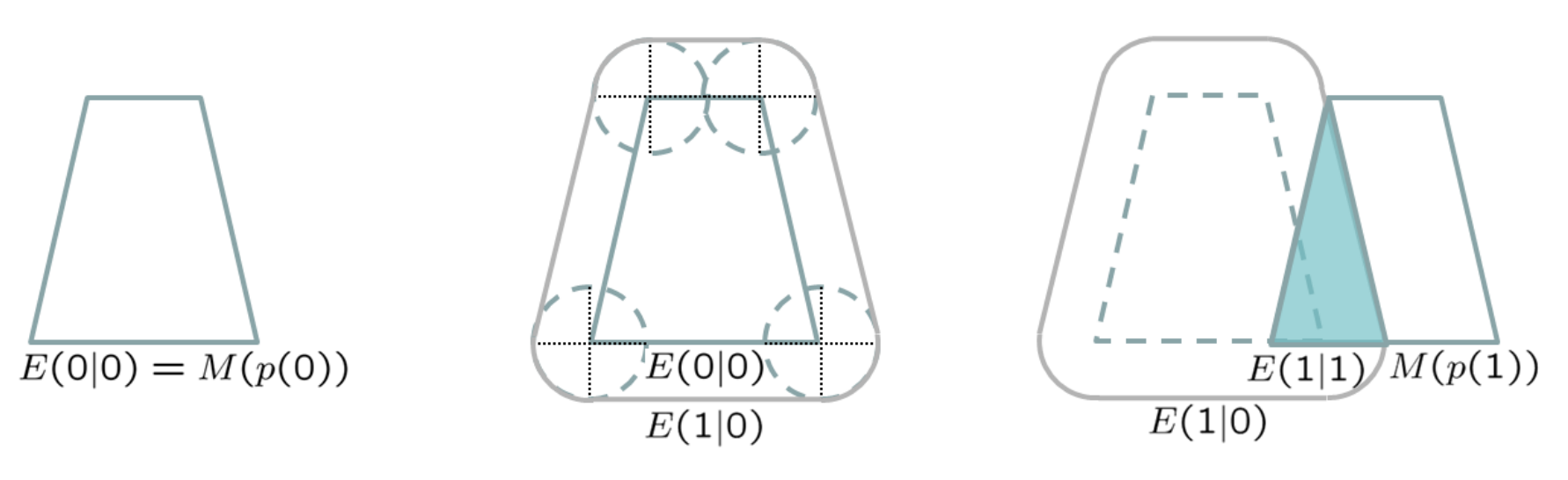}%
  \caption{Set-membership localization
    recursion~\eqref{eq:set_memb_recurs}: initialization, time update and
    measurement update.}
  \label{fig:set_membership_recursion}
\end{center}
\end{figure}

Running the recursion~\eqref{eq:set_memb_recurs} in its exact form is often
computationally intractable due to the large and increasing amount of data
required to describe the sets $E(t|t-1)$ and $E(t|t)$. To reduce the
computational complexity one typically over-approximates these sets with
bounded-complexity simple-structure sets, called \emph{approximating sets}.
For example, a common approximating set is the \emph{axis-aligned bounding
  box}, i.e., the smallest rectangle aligned with the reference axes
containing the set.  If $\Pi$ denotes the projection from the subsets of
$\real^2$ onto the collection of approximating sets, then the
recursion~\eqref{eq:set_memb_recurs} is rewritten as
\begin{subequations}
  \label{eq:set_memb_recurs_approx}
  \begin{align}
    E(0|0)    &= \Pi(M(p(0))),
    \label{eq:set_memb_recurs_approx-initialization}
    \\
    E(t|t-1)  &= \Pi( E(t-1|t-1) + \cball{\subscr{v}{max}}{0} ),
    \label{eq:set_memb_recurs_approx-time-update}
    \\
    E(t|t) &= \Pi( E(t|t-1) \intersection M(p(t)) ),
    \label{eq:set_memb_recurs_approx-measurement-update}
  \end{align}
\end{subequations}
where now the sets $E(t|t-1)$ and $E(t|t)$ are only approximations of the
feasible position sets.

\subsection{A centralized LP-based implementation: the eight half-planes
  algorithm}
In this section we propose a convenient choice of approximating sets for
set-membership localization and we discuss the corresponding time and
measurement updates. We begin with some preliminary notation.

We let $\HH_k$ be the set containing all possible collections of $k$
half-planes; in other words, an element of $\HH_k$ is a collection of $k$
half-planes.  Given an angle $\theta\in{[0,2\pi[}$ and a set of half-planes
$H = \{h_1, \ldots, h_k\}\in\HH_k$ with $h_i = \setdef{x\in\real^2}{a_i^T x
  \leq b_i, \norm{a_i}{}=1, b_i\in\real}$, define the linear program
$(H,\theta)$ by
\begin{equation} \label{eq:vtheta}
  \begin{split}
    \max   &\quad [\cos(\theta)\;\; \sin(\theta)] \cdot x\\
    \subject &\quad a_i^T x \leq b_i, \quad i\in\until{k}, \\
    & \quad \xmin\leq{x}\leq\xmax, \quad \ymin\leq{y}\leq\ymax.
  \end{split}
\end{equation}
As in Example~\ref{ex:abstract-framework-for-LP}, transcribe $(H,\theta)$
into an abstract optimization program $(H,\phi_\theta)$.  Recall that
$(H,\phi_\theta)$ is basis regular and has combinatorial dimension $2$, so
that its solution, i.e., the lexicographically-minimal minimum point of
$(H,\theta)$, is always a set of $2$ constraints, say
$\supind{h}{1}_{H,\theta}$ and $\supind{h}{2}_{H,\theta}$.  In other words,
the pair $\{\supind{h}{1}_{H,\theta},\supind{h}{2}_{H,\theta}\} \in \HH_2$
is computed as a function of an angle $\theta\in{[0,2\pi[}$ and of a
$k$-tuple $H=\{h_1,\ldots,h_k\}\in\HH_k$.

Now, as collection of approximating sets we consider the set $\HH_8$
containing the collections of $8$ half-planes. Note that the subset of
elements $\{h_1,\dots,h_8\} \in \HH_8$ such that $\intersection_{j=1}^8h_j$
is bounded is in bijection\footnote{Indeed, any convex polygon may be
  defined as the intersection of a finite number of half-planes; this
  definition is referred to as the H-representation of the convex polygon.}
with the set of convex polygons with at most $8$ edges.  Additionally, for
arbitrary $k$, we define the projection map $\map{\PiLP}{\HH_k}{\HH_8}$ as
follows: given $H\in\HH_k$, define $\PiLP(H)\in \HH_8$ to be the collection
of half-planes $\supind{h}{1}_{H,\theta}$ and $\supind{h}{2}_{H,\theta}$,
for $\theta \in \{0, \pi/2, \pi, 3\pi/2\}$. Note that our approximating set
$\PiLP(H)$ contains $H$ and is contained in the smallest axis-aligned
bounding box containing $H$; additionally, note that $\PiLP(H)$ contains
some possibly repeated half-planes because the same half-plane could be
part of the solution for distinct values of $\theta$.

Our definition of $\PiLP$ has the following interpretation: assuming the
target is known to satisfy all half-plane constraints in a set $H$, the
reduced-complexity possibly-unbounded polygon containing the target is
computed by solving four linear programs $(H,\theta)$, $\theta \in \{0,
\pi/2, \pi, 3\pi/2\}$; see
Figure~\ref{fig:target-localization-bounding-rectangle}.

\begin{figure}[h]
\begin{center}
  \includegraphics[width=.5\linewidth]{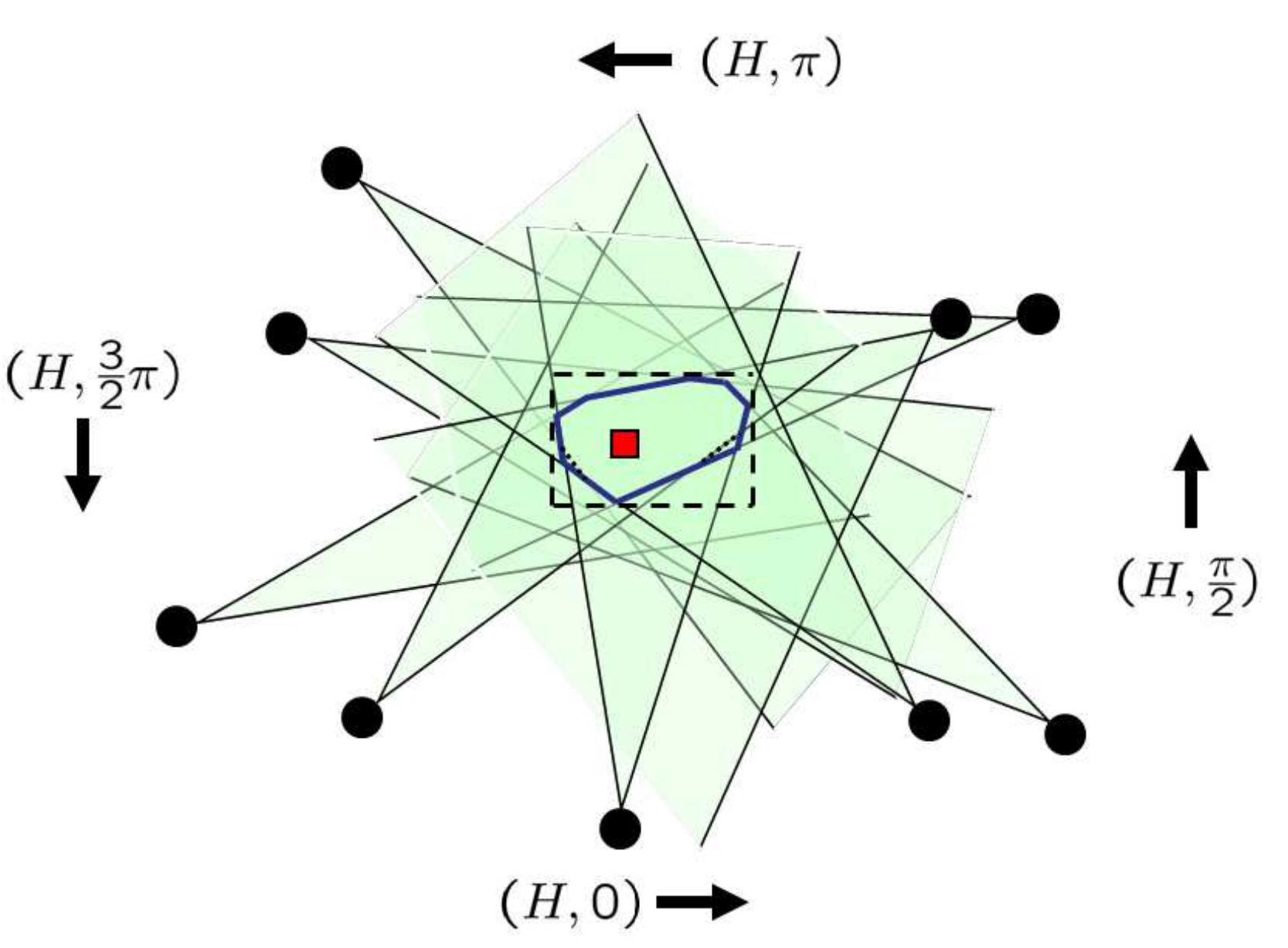}%
  \caption{The eight constraint half-planes (solid thick lines) are
    determined by the projection $\PiLP$. The dashed rectangle is the
    smallest axis-aligned bounding box containing the measurement set.}
  \label{fig:target-localization-bounding-rectangle}
\end{center}
\end{figure}

Finally, we review the approximated set-membership localization
recursion~\eqref{eq:set_memb_recurs_approx}.  We assume $\supind{h}{i}(t) =
\setdef{x\in\real^2}{\supind{a}{i}(t)^Tx\leq\supind{b}{i}(t),
  \norm{a_i}{}=1}$ is the half-plane containing the target measured by
sensor $i\in\until{n}$ at time $t\in\integernonnegative$.  The approximated
feasible position sets, elements of $\HH_8$, are
\begin{align*}
  {E}(t|t-1) &= \{  {h}_1(t|t-1),\dots,{h}_8(t|t-1)   \},
  \enspace \text{and}
  \\
  {E}(t|t)  &= \{  {h}_1(t|t),\dots,{h}_8(t|t)   \}.
\end{align*}

\smallskip
\noindent\emph{Initialization:}
Equation~\eqref{eq:set_memb_recurs_approx-initialization} reads
\begin{equation*}
  \{h_1(0|0),\dots,h_8(0|0)\}  = \PiLP\big(
  \{\supind{h}{1}(0),\dots,\supind{h}{n}(0)\} \big).
\end{equation*}

\smallskip
\noindent\emph{Time update:}
Assume $h_i(t|t) = \setdef{x\in\real^2}{\subind{a}{i}(t)^T
  x\leq\subind{b}{i}(t)}$, that is, $h_i(t|t)$ is characterized by the
coefficients $(\subind{a}{i}(t),\subind{b}{i}(t))$.  Since the target speed
satisfies $\norm{v}{}\leq\subscr{v}{max}$, at instant $t+\tau$ the target
is contained in the half-planes
\[
h_i(t+\tau|t) = \setdef{x\in\real^2}{\subind{a}{i}(t)^T
  x\leq\subind{b}{i}(t) + \subscr{v}{max} \tau}.
\]
Therefore, the time update consists in defining each $h_i(t+1|t)$ to be
$(\subind{a}{i}(t),\subind{b}{i}(t)+\subscr{v}{max})$; we refer to this
operation as to a \emph{time-translation} by an amount $\subscr{v}{max}$ of
the half-plane $h_i(t|t)$.  This time-update operation is equivalent to
equation~\eqref{eq:set_memb_recurs_approx-time-update} and does not
explicitly require an application of the projection $\PiLP$.

\smallskip
\noindent\emph{Measurement update:}
Equation~\eqref{eq:set_memb_recurs_approx-measurement-update} reads
\begin{gather*}
  \{h_1(t|t),\dots,h_8(t|t)\} = \PiLP(H(t)),
\end{gather*}
where the collection of constraints at time $t$ is
\begin{equation*}
  H(t) = \{h_1(t|t-1),\dots,h_8(t|t-1)\}\union
  \{\supind{h}{1}(t),\dots,\supind{h}{n}(t)\}.
\end{equation*}

\begin{remark}
  \begin{enumerate}
  \item During each iteration of the localization recursion, the time
    update step requires $8$ sums and the measurement update steps requires
    the solution of $4$ linear programs in $2$ variables and $n+8$
    constraints.

  \item Similar localization algorithms arise by selecting $\ell\geq3$ and
    by solving $\ell$ LP at each iteration parametrized by $\theta\in
    \{0,2\pi/\ell,\dots,2(\ell-1)\pi/\ell\}$. Larger values of $\ell$ lead
    to tighter approximating polygons. \oprocend

  \end{enumerate}
\end{remark}

\subsection{A distributed eight half-planes algorithm}
We consider a scenario in which the sensors measuring the target position
also have computation and communication capabilities so that they form a
synchronous network as described in Section~\ref{sec:network-modeling}.
Let $\subscr{\GG}{sn} = (\until{n}, \subscr{E}{sn})$ be the undirected
communication graph among the sensors $\until{n}$; assume $\GG$ is
connected.
Assume the sensors communicate at each time $t\in\integernonnegative$ and
perform measurements of the target at unspecified times in
$\integernonnegative$ (communication takes place at higher rate than
sensing). For simplicity, we assume the first measurement at each node
happens at time $0$.

We aim to design a distributed algorithm for the sensor network to localize
a moving target. The idea is to run \emph{local set-membership recursions}
(with time and measurement updates) at each node while exchanging
constraints in order to achieve constraints consensus on a set-membership
estimate.  Distributed constraint re-examination is obtained as follows: at
each time, each node keeps in memory the last $m$ measurements it took and,
after an appropriate time-update, re-introduces them into the $\PiLP$
computation.  We begin with an informal description.
\begin{quote}
  \emph{Eight Half-Planes Consensus Algorithm}: The processor state at each
  processor $i$ contains a set $\Hoptimal{i}$ of $8$ candidate optimal
  constraints and a set $\Hmeasured{i}$ containing the last $m$
  measurements, for some $m>0$. These sets are initialized to the first
  sensor measurement.  At each communication round, the processor performs
  the following tasks: (i) it transmits $\Hoptimal{i}$ to its out-neighbors
  and acquires from its in-neighbors their candidate constraints; (ii) it
  performs a time-update, that is, a time-translation by an amount $\vmax$,
  of all candidate optimal, measured and received constraints; (iii) it
  updates the set of measured constraints if a new measurement is taken;
  and (iv) it updates $\Hoptimal{i}$ to be the projection $\PiLP$ of all
  candidate optimal, measured and received constraints.
\end{quote}

Next we give a pseudo-code description.

\medskip \hrule width \linewidth \smallskip

\noindent\begin{minipage}{0.44\linewidth}\textbf{\texttt{Problem data:}}%
\end{minipage}%
\begin{minipage}[t]{0.56\linewidth}A network $\subscr{\GG}{sn}$ of
  sensors that measure half-plane constraints%
\end{minipage}
\vspace{.05em}

\noindent\begin{minipage}{0.44\linewidth}\textbf{\texttt{Algorithm:}}%
\end{minipage}%
\begin{minipage}{0.56\linewidth}Eight Half-Planes Consensus%
\end{minipage}

\noindent\begin{minipage}{0.44\linewidth}\textbf{\texttt{Message
      alphabet:}}%
\end{minipage}%
\begin{minipage}{0.56\linewidth}$\alphabet = \HH_8 \union \{\nll\}$%
\end{minipage}

\noindent\begin{minipage}{0.44\linewidth}\textbf{\texttt{Processor state:}}%
\end{minipage}%
\begin{minipage}[t]{0.56\linewidth}
  $\Hoptimal{i} \in \HH_8$\\
  $\Hmeasured{i} \in \HH_m$ for some $m>0$
\end{minipage}
\vspace{.05em}

\noindent\begin{minipage}{0.44\linewidth}\textbf{\texttt{Initialization:}}%
\end{minipage}%
\begin{minipage}[t]{0.56\linewidth}
  $\Hoptimal{i} := \{\supind{h}{i}(0),\dots,\supind{h}{i}(0)\}$\\
  $\Hmeasured{i} := \{\supind{h}{i}(0),\dots,\supind{h}{i}(0)\}$%
\end{minipage}

\bigskip

\noindent\textbf{\texttt{function}} $\msg\big( (\Hoptimal{i},
\Hmeasured{i}), j\big)$
\begin{algorithmic}[1]
  \STATE \textbf{return} $\Hoptimal{i}$
\end{algorithmic}

\medskip

\noindent\textbf{\texttt{function}}  $\stf\big((\Hoptimal{i}, \Hmeasured{i}), y\big)$ \\
\emph{\% executed by node~$i$, with $y_j := \Hoptimal{j}$}\\[-1.9em]
\begin{algorithmic}[1]

  \STATE time-translate by an amount $\vmax$  all constraints
  constraints in     $\Hmeasured{i}$, $\Hoptimal{i}$, and
  $\union_{j\in \innbrs(i)} y_j$

  \IF{a new measurement is taken at this time,}
  \STATE add it to $\Hmeasured{i}$; drop oldest
  measurement from $\Hmeasured{i}$
  \ENDIF

  \STATE set $\Hoptimal{i} := \PiLP \Big( \Hmeasured{i} \union
  \Hoptimal{i} \union_{j\in \innbrs(i)} y_j \Big)$

  \STATE \textbf{return} $(\Hoptimal{i},\Hmeasured{i})$
\end{algorithmic}

\smallskip \hrule width \linewidth \medskip

Finally, we collect some straightforward facts about this algorithm; we
omit the proof in the interest of brevity.

\begin{proposition}\textbf{\textup{(Properties of the eight half-planes consensus algorithm)}}
  Consider a connected network $\subscr{\GG}{sn}$ of sensors that measure
  half-plane constraints and that implement the eight half-plane consensus
  algorithm.  Assume the target does not move, that is, set $\vmax=0$. The
  following statements hold:
  \begin{enumerate}
  \item the candidate optimal constraints at each node contain the target
    at each instant of time;
  \item the candidate optimal constraints at each node monotonically
    improve over time; and
  \item additionally, if each node makes at most $m$ measurements in finite
    time, then the candidate optimal constraints at each node converge in
    finite time to the globally optimal $8$ half-plane constraints.
  \end{enumerate}
\end{proposition}

Next, let us state some memory complexity bounds for the eight half-plane
consensus algorithm. Again, we adopt the convention that a memory unit is
the amount of memory required to store a constraint in $H$.  Each node $i$
of the network requires $(8+m+8\card(\innbrs(i)))$ memory units in order to
implement the algorithm.
If we assume that number $m$ of stored measurements is independent of $n$
and that the indegree of each node is bounded irrespectively of $n$, then
the algorithm memory complexity is in $O(1)$. Vice-versa, in worst-case
graphs, the algorithm memory complexity is in $O(n)$.

\section{Application to formation control for robotic networks}
\label{sec:mintime-formation}

In this section we apply constraints consensus ideas to formation control
problems for networks of mobile robots.  We focus on formations with the
shapes of a point, a line, or a circle. (The problem of formation control
to a point is usually referred to as the rendezvous or gathering problem.)
We solve these formation control problems in a time-efficient manner via a
distributed algorithm regulating the communication among robots and the
motion of each robot.

\subsection{Model of robotic network}
We define a robotic network as follows. Each robot is equipped with a
processor and robots exchange information via a communication graph.
Therefore, the group of robots has the features of a synchronous network
and can implement distributed algorithms as defined in
Section~\ref{sec:network-modeling}. However, as compared with a synchronous
network, a robotic network has two distinctions: (i) robots control their
motion in space, and (ii) the communication graph among the robots depends
upon the robots positions, rather than time.

Specifically, the \emph{robotic network} evolves according to the following
discrete-time communication, computation and motion model. Each robot
$i\in\until{n}$ moves between rounds according to the first order
discrete-time integrator $\supind{p}{i}(t+1) = \supind{p}{i}(t) +
\supind{u}{i}(t)$, where $\supind{p}{i}\in\real^2$ and
$\|\supind{u}{i}\|^2\leq \umax>0$.  At each discrete time instant, robots
at positions $P_n$ communicate according to the disk graph
$\subscr{\GG}{disk}(P_n) = (\until{n}, \subscr{E}{disk}(P_n))$ defined as
follows: an edge $(i,j)\in\until{n}^2$, $i\ne{j}$, belongs to
$\subscr{E}{disk}(P_n)$ if and only if
$\norm{\supind{p}{i}-\supind{p}{j}}{}\leq \rcmm$ for some $\rcmm>0$.

A \emph{distributed algorithm for a robotic network} consists of (1) a
distributed algorithm for a synchronous network, that is, a processor
state, a message alphabet, a message-generation and a state-transition
function, as described in Section~\ref{sec:network-modeling}, (2) an
additional function, called the \emph{control function}, that determines
the robot motion, with the following domain and co-domain:
\begin{equation*}
  \map{\ctrl} {\real^2  \times W \times \alphabet^n} {\cball{\umax}{0} }.
\end{equation*}
Additionally, we here allow the message generation and the state transition
to depend upon not only the processor state but also the robot position.
The state of the robotic network evolves as follows. First, at each
communication round $t$, each processor $i$ sends to its outgoing neighbors
a message computed by applying the message-generation function to the
current values of $\supind{p}{i}$ and $\supind{w}{i}$.  After a negligible
period of time, the $i$th processor resets the value of its processor state
$\supind{w}{i}$ by applying the state-transition function to the current
values of $\supind{p}{i}$ and $\supind{w}{i}$, and to the messages received
at time $t$. Finally, the position $\supind{p}{i}$ of the $i$th robot at
time $t$ is determined by applying the control function to the current
value of $\supind{p}{i}$ and $\supind{w}{i}$, and to the messages received
at time $t$.

In formal terms, if $\supind{y}{i}(t)$ denotes the message vector received
at time $t$ by agent $i$ (with $\supind{y}{i}_j(t)$ being the message
received from agent $j$), then the evolution is determined by
\begin{equation*}
  \begin{split}
    \supind{y}{i}_j(t) &= \msg(p^{[j]}(t-1),w^{[j]}(t-1),i),\\
    \supind{w}{i}(t) &
    = \stf(p^{[i]}(t-1),\supind{w}{i}(t-1), y^{[i]}(t)), \\
    p^{[i]}(t) &= p^{[i]}(t-1) +
    \ctrl(p^{[i]}(t-1), \supind{w}{i}(t), y^{[i]}(t)),
  \end{split}
\end{equation*}
with the convention that $\msg(p^{[j]}(t-1),w^{[j]}(t-1),i)=\nll$ if
$(i,j)\not\in \subscr{E}{disk}(p^{[1]}(t-1),\dots,p^{[n]}(t-1))$.

\subsection{Formation tasks and related optimization problems}
\label{sec:mintime-formation_optimal-formation}
Numerous definitions of robot formation are considered in the multi-agent
literature. Here we consider a somehow specific situation.  Let
$\subscr{S}{points}$, $\subscr{S}{lines}$, and $\subscr{S}{circles}$ be the
set of points, lines and circles in the plane, respectively. We refer to
these three sets as the \emph{shape sets}.  We aim to lead all robots in a
network to a single element of one of the shape sets If $S$ is a selected
shape set, the \emph{formation task} is achieved by the robotic network if
there exists a time $T\in\natural$ such that for all $t\geq T$, all robots
$i\in\until{n}$ satisfy $\supind{p}{i}(t) \in s$ for some element $s\in S$.
Specifically, the \emph{point-formation, or rendezvous task} requires all
connected robots to be at the same position, the \emph{line-formation task}
requires all connected robots to be on the same line, and the
\emph{circle-formation task} requires all connected robots to be on the
same circle.

We are interested in distributed algorithms that achieve such formation
tasks optimally with respect to a suitable cost function. For the
point-formation and line-formation tasks, we aim to minimize completion
time, i.e., the time required by all robots to reach a common shape.  For
the circle-formation task, we aim to minimize the product between the time
required to reach a common circle, and the diameter of the common circle.

\begin{remark}[Circle formation] For the circle-formation problem we do not
  select the completion time as cost function, because of the following
  reasons. The centralized version of the minimum time circle-formation
  problem is equivalent to finding the minimum-width annulus containing the
  point-set. For arbitrary data sets, the minimum-width annulus has
  arbitrarily large minimum radius and bears similarities with the solution
  to the smallest stripe problem.  For some configurations, all points are
  contained only in a small fraction of the minimum-width annulus; this is
  not the solution we envision when we consider moving robots in a circle
  formation.  Therefore, we consider, instead, the smallest-area
  annulus. This cost function penalizes both the difference of the radiuses
  of the annulus (width of the annulus) and their sum.  \oprocend
\end{remark}

The key property of the \emph{minimum-time point-formation task},
\emph{minimum-time line-formation task}, and \emph{optimum circle-formation
  task} is that their centralized versions are equivalent to finding the
smallest ball, stripe and annulus, respectively, enclosing the $n$ agents'
initial positions. We state these equivalences in the following lemma
without proof.

\begin{lemma}[Optimal shapes from geometric optimization]
  Given a set of distinct points $\{p_1,\dots,p_n\} \subset \real^2$,
  consider the three optimization problems:
  \begin{equation*}
  \begin{split}
    \min_{p\in\subscr{S}{points}} \; \max_{j\in\until{n}} & \; \norm{p_j - p}{},\\
    \min_{\ell\in\subscr{S}{lines}}\; \max_{j\in\until{n}} &\; \dist(p_j, \ell), \\
    \min_{c\in\subscr{S}{circles}} \; \max_{j\in\until{n}} &\; \dist(p_j, c) \cdot
    \radius(c),
  \end{split}
  \end{equation*}
  where $\radius(c)$ denotes the radius of the circle $c$.  These three
  optimization problems are equivalent to the smallest enclosing ball, the
  smallest enclosing stripe (for points in stripe-generic position), and
  the smallest enclosing annulus problem, respectively.

  Therefore, they are abstract optimization problems with combinatorial
  dimension $3$, $5$ and $4$, respectively.
\end{lemma}

We conclude this section with some useful notation.  We let
$\texttt{target\_set}(\{p_1,\dots,p_n\})$ denote the point, the line or the
circle equidistant from the boundary of the smallest enclosing ball, stripe
or annulus, respectively.

\subsection{Connectivity assumption, objective and strategy}
We assume that the robotic network is connected at initial time, i.e., that
the graph $\subscr{\GG}{disk}(P_n(0))$ is connected, and we aim to achieve
the formation task while guaranteeing that the state-dependent
communication graph remains connected during the evolution.  The following
connectivity maintenance strategy was originally proposed
in~\cite{HA-YO-IS-MY:99} and a comprehensive discussion is
in~\cite{FB-JC-SM:09}.  The key idea is to restrict the allowable motion of
each robot so as to preserve the existing edges in the communication graph.
We present this idea in three steps.

First, in a network with communication edges $\subscr{E}{disk}$, if agents
$i$ and $j$ are neighbors at time $t\in\naturalzero$, then we require that
their positions at time $t+1$ belong to
\begin{equation*}
  \XX(\supind{p}{i}(t), \supind{p}{j}(t)) =
  \Bigcball{\half\rcmm}{\frac{\supind{p}{i}(t)+\supind{p}{j}(t)}{2}}.
\end{equation*}
If all neighbors of agent $i$ at time $t$ are at locations
$\supind{Q}{i}(t) = \{q_1,\dots,q_l\}$, then the (convex) \emph{constraint
  set} of agent $i$ is
\begin{equation*}
    \XX(\supind{p}{i}(t), \supind{Q}{i}(t))
    = \bigcap_{q\in \{q_1,\dots,q_l\}}
    \Bigcball{\half\rcmm}{\frac{\supind{p}{i}(t)+q}{2}}.
\end{equation*}

Second, given $p$ and $q$ in $\real^2$ and a convex closed set $Q \subset
\real^2$ with $p \in Q$, we introduce the \emph{from-to-inside} function,
denoted by $\fti$ and illustrated in Figure~\ref{fig:fti}, that computes
the point in the closed segment $[p, q]$ which is at the same time closest
to $q$ and inside $Q$. Formally,
\begin{equation*}
    \fti(p, q,Q) =
    \begin{cases}
        q, &\text{if } q\in Q,\\
        [p,q] \intersection \partial Q, \quad&  \text{if } q\notin Q.
    \end{cases}
\end{equation*}

\begin{figure}[h]
\begin{center}
  \includegraphics[width=.28\linewidth]{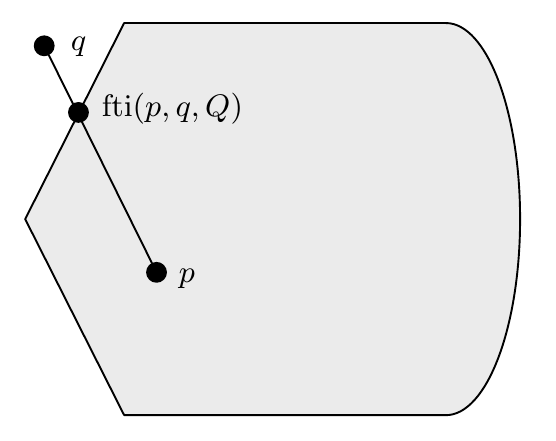}%
  \includegraphics[width=.28\linewidth]{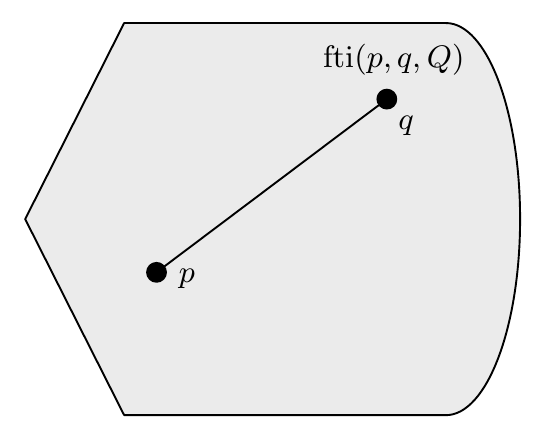}%
  \caption{Illustration of the $\fti$ function; courtesy of the authors
    of~\cite{FB-JC-SM:09}.}
  \label{fig:fti}
\end{center}
\end{figure}

Third and final, we assume that, independent of whatever control algorithm
dictates the robots evolution, if $\supind{p}{i}_{\text{target}}(t)$ denote
the desired target positions of agent $i$ at time $t+1$, then we allow
robot $i$ to move towards that location only so far as the constraint set
allows. This is encoded by:
\begin{equation*}
  \supind{p}{i}(t+1) = \fti(\supind{p}{i}(t),
  \supind{p}{i}_{\text{target}}(t), \XX(\supind{p}{i}(t),
  \supind{Q}{i}(t))) .
\end{equation*}

\subsection{Move-to-consensus-shape strategy}
The minimum-time point-formation, minimum-time line-formation, and optimum
circle formation tasks appear intractable in their general form due to the
state-dependent communication constraints.  To attack these problem, we
search for an efficient strategy that converges to the optimal one when the
upper bound on the robot speed $\umax$ goes to zero or, in other words,
when information transmission tends to be infinitely faster than motion.
To design such a strategy, we aim to reach consensus on the centralized
solution to the problem, i.e., the optimal shape, and use the solution as a
reference for the agents motion.

A simple sequential solution is as follows: first, the agents compute the
optimal shape via constraints consensus, and then, when consensus is
achieved, they move toward the closest point in the target shape (point,
line or circle).  An improved strategy allows concurrent execution of
constraints consensus and motion: while the constraints consensus algorithm
is running, each agent moves toward the estimated target position while
maintaining connectivity of the communication graph. We first provide an
informal description.

\begin{quote}
  \emph{Move-to-consensus-shape strategy}: The processor state at each
  robot $i$ consists of a set $\supind{B}{i}$ of $\delta$ candidate optimal
  constraints and a binary variable $\supind{halt}{i}\in\{0,1\}$.  The set
  $\supind{B}{i}$ is initialized to $\supind{p}{i}(0)$ and
  $\supind{halt}{i}$ is initialized to $0$. At each communication round,
  the processor performs the following tasks: (i) it transmits
  $\supind{p}{i}$ and $\supind{B}{i}$ to its neighbors and acquires its
  neighbors' candidate constraints and their current position; (ii) it runs
  an instance of the constraints consensus algorithm for the geometric
  optimization program of interest (smallest enclosing ball, stripe or
  annulus); if the constraints consensus halting condition is satisfied, it
  sets $\supind{halt}{i}$ to $1$; (iii) it computes a robot target position
  based on the current estimate of the optimal shape; (iv) it moves the
  robot towards the target position while respect input constraint and, if
  $\supind{halt}{i}$ is still zero, enforcing connectivity with its current
  neighbors.
\end{quote}

Next we give a pseudo-code description.  We let
$P_n(0)=\{\supind{p}{1}(0),\ldots,\supind{p}{n}(0)\}$.

\bigskip\bigskip
\bigskip \hrule width \linewidth \smallskip

\noindent\begin{minipage}{0.44\linewidth}\textbf{\texttt{Problem data:}}%
\end{minipage}%
\begin{minipage}[t]{0.56\linewidth} A robotic network and a shape set%
\end{minipage}

\noindent\begin{minipage}{0.44\linewidth}\textbf{\texttt{Algorithm:}}%
\end{minipage}%
\begin{minipage}{0.56\linewidth}Move-to-Consensus-Shape%
\end{minipage}

\noindent\begin{minipage}{0.44\linewidth}\textbf{\texttt{Message
      alphabet:}}%
\end{minipage}%
\begin{minipage}{0.56\linewidth}$\alphabet = \real^2 \union (\real^2)^\delta \union \{\nll\}$%
\end{minipage}

\noindent\begin{minipage}{0.44\linewidth}\textbf{\texttt{Processor
      state:}}%
\end{minipage}%
\begin{minipage}[t]{0.56\linewidth}
  $\supind{B}{i} \subset P_n(0)$ with $\card(\supind{B}{i})=\delta$\\%
  $\supind{\texttt{halt}}{i}\in\{0,1\}$
\end{minipage}

\noindent\begin{minipage}{0.44\linewidth}\textbf{\texttt{Physical
      state:}}%
\end{minipage}%
\begin{minipage}[t]{0.56\linewidth}
  $\supind{p}{i} \in \real^2$
\end{minipage}

\smallskip
\noindent\begin{minipage}{0.44\linewidth}\textbf{\texttt{Initialization:}}%
\end{minipage}%
\begin{minipage}[t]{0.56\linewidth}
  $\supind{B}{i} := \{\supind{p}{i}(0),\dots,\supind{p}{i}(0)\}$\\
  $\supind{\texttt{halt}}{i}:=0$
\end{minipage}

\bigskip

\noindent\textbf{\texttt{function}}  $\msg\big(\supind{p}{i},\supind{B}{i},j\big)$ \begin{algorithmic}[1]
  \STATE \textbf{return} $(\supind{p}{i},\supind{B}{i})$
\end{algorithmic}

\medskip

\noindent\textbf{\texttt{function}}  $\stf\big(\supind{B}{i},
\supind{\texttt{halt}}{i}, y\big)$ \\
\emph{\% executed by node~$i$, with $y_j=(\supind{p}{j},\supind{B}{j})
  := \msg(\supind{B}{j},i)$}\\[-1.9em]
\begin{algorithmic}[1]
  \STATE $\subscr{S}{tmp} := \{\supind{p}{i}(0)\} \union \supind{B}{i} \union \big(
  \union_{j\in \innbrs(i)} \supind{B}{j} \big)$
  \STATE $\supind{B}{i} := \subexLP( \subscr{S}{tmp}, \supind{B}{i})$
  \STATE \textbf{if} $\supind{B}{i}$ has not changed for $2n$ rounds,
  \STATE \quad \textbf{then}  $\supind{\texttt{halt}}{i}:=1;$ \quad \textbf{end if}
  \STATE \textbf{return} $(\supind{B}{i}, \supind{\texttt{halt}}{i})$
\end{algorithmic}
\medskip

\noindent\textbf{\texttt{function}}  $\ctrl\big(\supind{B}{i}, \supind{p}{i}, y\big)$ \\
\emph{\% executed by node~$i$, with $y_j=(\supind{p}{j},\supind{B}{j})
  := \msg(\supind{B}{j},i)$}\\
\emph{\% map \texttt{target\_set} defined in
  Section~\ref{sec:mintime-formation_optimal-formation}}\\[-1.9em]
\begin{algorithmic}[1]
  \STATE $\subscr{S}{\texttt{trgt}} := \texttt{target\_set}(\supind{B}{i})$\\
  \STATE $\displaystyle \subscr{p}{trgt} := \arg \min_{p \in \subscr{S}{\texttt{trgt}}}
  \norm{\supind{p}{i} - p}{}$
  \STATE \textbf{if} $\supind{\texttt{halt}}{i}=0$, \textbf{then}\\
  $\subscr{\XX}{cnstr}
  := \bigcap_{j\in \innbrs(i)}
  \Bigcball{\half\rcmm}{\frac{\supind{p}{i}+\supind{p}{j}}{2}}
  \bigcap \bigcball{\umax}{\supind{p}{i}}$\\
  \textbf{else} $\subscr{\XX}{cnstr}
  := \bigcball{\umax}{\supind{p}{i}};$ \quad \textbf{end if}
  \STATE \textbf{return} $\fti(\supind{p}{i}, \subscr{p}{trgt}, \subscr{\XX}{cnstr})-\supind{p}{i}$
\end{algorithmic}

\smallskip \hrule width \linewidth \medskip

We refer the interested reader to~\cite{GN-FB:06d} for numerical simulation
results for the move-to-consensus-shape strategy.  Finally, we state the
correctness of the algorithm and omit the proof in the interest of brevity.

\begin{proposition}
  \textbf{\textup{(Properties of the move-to-consensus-shape algorithm)}}
  On a robotic network with communication graph $\subscr{G}{disk}$ and
  bounded control inputs $\umax$, the move-to-consensus-shape strategy
  achieves the desired formation control tasks.

  In the limit as $\umax\to0^+$, the move-to-consensus-shape strategy
  solves the optimal formation control tasks, i.e., the minimum-time
  point-formation, minimum-time line-formation, and optimum circle
  formation tasks.
\end{proposition}




\section{Conclusions}

We have introduced a novel class of distributed optimization problems and
proposed a simple and intuitive algorithmic approach. Additionally, we have
rigorously established the correctness of the proposed algorithms and
presented a thorough Monte Carlo analysis to substantiate our conjecture
that, for a broad variety of settings, the time complexity of our
algorithms is linear. Finally, we have discussed in detail two modern
applications: target localization in sensor networks and formation control
in robotic networks.

Promising avenues for further research include proving the
linear-time-complexity conjecture, as well as applying our algorithms to
(i) optimization problems for randomly switching graphs and gossip
communication, (ii) distributed machine learning
problems~\cite{JB-YD-JT-OW:08,YL-VR-LV:08,KF-BBL-PT:06}, and (iii) convex
quadratic programming~\cite{RG:99}. Additionally, it is of interest to
verify the performance of our proposed target localization and formation
control algorithms in experimental setups.

\appendix

\subsection*{Abstract framework for the smallest enclosing stripe problem}
In this appendix we consider the smallest enclosing stripe problem in
Example~\ref{rem:examples} for stripe-generic points and show that is
satisfies the abstract optimization axioms. We begin with some notation.
Let $H$ be the set of constraints, i.e., the set of points in the plane for
which we aim to compute the smallest enclosing stripe. With a slight abuse
of notation, $h\in H$ denotes both a constraint and a point depending on
the context.  Let $\phi(G)$, $G\subset H$, be the width of the smallest
stripe enclosing the points in $G$. A constraint $h$ is violated by the set
$G$, i.e., $\phi(G\union \{h\})>\phi(G)$, if the point $h$ does not belong
to the smallest stripe enclosing $G$.

First, note that three points are necessary, but in general not sufficient,
to uniquely identify the smallest enclosing stripe of a set $H$ of three or
more points. Indeed, the boundary of the smallest stripe enclosing $H$
contains at least three points of $H$, e.g., point $h_1$ on one line and
points $\{h_2,h_3\}$ on the other line. So, the smallest enclosing stripe
identifies the triplet $\{h_1,h_2,h_3\}$. However, a simple geometric
argument shows that the converse is not true, i.e., three points are not
sufficient. Three points would be sufficient if the triplet was an ordered
set and uniquely identified a stripe.
To uniquely identify the smallest enclosing stripe for a set $H$ of five or
more points, one needs to add to the triplet $\{h_1,h_2,h_3\}$ another two
points of $H$ that belong to the correctly stripe among the three stripes
determined by $\{h_1,h_2,h_3\}$.  In summary, this discussion shows that
the combinatorial dimension is $5$.

Second, we prove that the two axioms of abstract optimization are
satisfied.  As usual, monotonicity is trivially satisfied, thus we need to
prove locality.
Suppose, by contradiction, that locality does not hold. Therefore there
exist $F$, $G$ and $h$, $F\subset G\subset H$ and $h\in H$, with
$-\infty<\phi(F) = \phi(G)$ such that $\phi(G \union \{h\})>\phi(G)$ and
$\phi(F\union\{h\}) = \phi(F)$.
Let $S_G$ and $S_F$ be the smallest stripes for $G$ and $F$, so that
$(\subscr{g}{1},\subscr{g}{2},\subscr{g}{3})$ and
$(\subscr{f}{1},\subscr{f}{2},\subscr{f}{3})$ are the ordered triplets of
points defining $S_G$ and $S_F$, respectively.
Since $h$ is violated by $G$ but not by $F$, this means that $h$ belongs to
$S_F$ but not to $S_G$. Therefore, the stripes $S_G$ and $S_F$ must be
different, i.e., $(\subscr{g}{1},\subscr{g}{2},\subscr{g}{3}) \neq
(\subscr{f}{1},\subscr{f}{2},\subscr{f}{3})$.
Because the points are in stripe-generic positions, we know that
$\dist(\subscr{g}{1},l(\subscr{g}{2},\subscr{g}{3})) \neq
\dist(\subscr{f}{1},l(\subscr{f}{2},\subscr{f}{3}))$ and, therefore,
$\phi(G)>\phi(F)$. This proves the contradiction.

\end{document}